\documentclass[11pt,a4paper,reqno]{amsart}
\usepackage{amsaddr}
\usepackage{amsthm,amssymb,amsmath,graphicx,enumerate}
\usepackage{mathrsfs,setspace,pstricks,multicol,latexsym}
\usepackage{epsfig, subcaption, graphics}
\usepackage{hyperref}
\usepackage{dsfont}
\usepackage{multirow}
\usepackage{float}
\usepackage{caption}
\usepackage{array}

\numberwithin{equation}{section}
\newtheorem{thm}{Theorem}[section]
\newtheorem{definition}[thm]{Definition}
\newtheorem{lem}[thm]{Lemma}

\newtheorem{rmk}[thm]{Remark}

\allowdisplaybreaks

\newcommand{\toprule}{\hline}
\newcommand{\bottomrule}{\hline}
\newcommand{\midrule}{\hline}

\begin{document}
\author{Anindya Goswami}
\address{Department of Mathematics, IISER Pune, India}
\email{anindya@iiserpune.ac.in }
\author{Nimit Rana}
\address{Department of Mathematics, University of York, UK}
\email{nimit.rana@york.ac.uk}

\title{A market resilient data-driven approach to option pricing}\thanks{The authorship is distributed equally among all authors. The authors' names appear in the alphabetical order of their surnames. }

\begin{abstract}
In this paper, we present a data-driven ensemble approach for option price prediction whose derivation is based on the no-arbitrage theory of option pricing. Using the theoretical treatment, we derive a common representation space for achieving domain adaptation. The success of an implementation of this idea is shown using some real data. Then we report several experimental results for critically examining the performance of the derived pricing models.
\end{abstract}

\maketitle

\section{Introduction}
Some products or services are often bought to protect the buyer from uncertainties. Warranty is one of them. Figuring out a fair price for a warranty is not an easy task. Since the financial markets consist of risky assets, traders purchase some warranty-type contracts, called options. Several different types of options are traded by numerous traders daily in every modern stock exchange. Determination of an option's fair price is one of the central questions in Mathematical Finance. Following the seminal work \cite{BS} by Black and Scholes the mathematical treatment for addressing the option pricing problem quickly flourished. In the last fifty years, thousands of peer-reviewed articles came into existence to address various mathematical, computational, or statistical aspects of fair pricing of a large spectrum of options under diverse scenarios. Consideration of increasingly realistic stochastic models of asset returns and then showing that the model could preclude arbitrage, before deriving a stochastic or PDE representation of an option's price is the state of the art for option pricing research. These works have tremendously enriched our understanding of fair pricing of options, beyond any doubt. Besides, option pricing research also played a vital role in inspiring deeper and more diverse investigation of stochastic calculus, differential equations, numerical simulations etc. Given this, it is natural to have a judgment that ``pricing of option is not just about data science''. Without disagreeing with that judgment, we propose a novel direction of research that translates some existing data science ideas into the language of stochastic modeling to derive some theoretical algorithms of data-driven option pricing that leverage from both worlds. 

A solution to the option pricing problem is called data-driven if the approach assumes no theoretical model of asset dynamics and relies only on the observed data. However, it is possible to utilize the general theory of fair pricing of options on an abstract underlying asset while deriving a data-driven option pricing model. In other words, despite the pricing solution being data-driven the derivation of the solution need not always be data-driven. For example, in a recent work \cite{CLZ2021}, some theoretical constraints of fair price of options are utilized to come up with a data-driven model. The investigations on the prospect of various types of data-driven option pricing are not new in the literature. Nevertheless, the volume of work in this direction is much less than its counterpart in the theoretical option pricing. A comprehensive literature survey is beyond the scope of this paper.  In the initial works of \cite{Malliaris1993, Hutch94, Maddala96} the promise of data-driven option pricing is reported using the data of option contracts on the S\&P 500. Further investigation on feature selection appeared in \cite{Boek95, Malliaris1996}. Modularity is used in \cite{Gradojevic2009,Yao2000} for improving the pricing quality. In the context of option pricing, modularity entails dividing the data set into disjoint modules, according to the moneyness and time-to-maturity parameters of the contracts. Some more recent works \cite{GRT, SVS24, KTP23} build upon the above findings and a theoretical property of options contracts, called homogeneity hint \cite{GARCIA200093}. 

In this paper, we first revisit the homogeneity hint property by producing theoretical proof in a wide general setting. Then we explain the shortcomings of the option price prediction models based only on the homogeneity hint. Such models can learn from one asset class data and be applied to another, provided their log returns have similar distributions under their risk-neutral measures. However, requiring this identity in two different assets is quite restrictive. For this reason, we formulate, in a parametric setting, a common representation space that bridges two different risk-neutral return distributions. The common representation space is used for connecting option prices on assets having no similarities. This is a fresh idea for the machine learning application of option pricing, where domain adaptation is implemented. We show theoretically, as well as empirically that a data-driven approach based on a common representation space can gain accuracy on test data having significant domain shifts. For the empirical study, we consider the option price data of two indices in the National Stock Exchange (NSE) of India. It is worth noting that options on these indices (NIFTY 50, and NIFTY Bank) have very high trading volume in respect of the global market. 

We compare the performance of both approaches, the one that is based on homogeneity hint (HH) and the one that accounts for domain shift (DS) using the price data of options on NIFTY 50, and NIFTY Bank indices. We identify that option price data during the COVID lockdown period in India has a significant domain shift and is hence eligible to assess the performance of models for domain shift. The assessment shows a significant gain in accuracy for models that utilize a common representation space.
We also propose an ensemble model using the models based on HH and the one that is to address DS. We put forward a notion of domain shift quotient, that measures the degree of domain shift and is utilized to set weights on constituent models in the ensemble model. We have performed extensive experiments to assess the performance of the ensemble model using empirical as well as synthetic data. In this paper, we also assess the performance of each of these three approaches for multi-source training models, where the model is trained on data coming from symbols NIFTY 50 and NIFTY Bank. We report empirical evidence of the reliability of the multi-source ensemble model. Option price prediction model with multi-source training is capable of giving reliable answers even for an underlying asset, which has no or little historical data. This overcomes the common hurdle faced by any data-driven approach.

It is worth mentioning that we include no macroeconomic or fundamental financial variables in the feature sets, to retain the interpretability of the model. Thus the proposed approaches do not lack the interpretability of the model which is a primary criticism that a data science model commonly receives. We also clarify that the goal of this paper is not merely to showcase the usefulness of an existing or newly developed machine-learning method. This rather lays the foundation of the theoretical solution to domain adaptation, a data science idea, in the context of option pricing. 
The full potential of a non-asset-specific model with the above-mentioned domain adaptation may further be explored by extensive experimentation. In this connection, we recall that \cite{MR3991069} reports a similar extensive experiment to study some other universal non-asset-specific relations captured by a deep learning model. However, keeping in mind that this is the first research on domain adaptation in the context of option pricing, we focus more on theory building than on large-scale empirical verification across all markets. 

In this paper, we introduce a variable called \emph{volatility scalar} that plays a key role in the formulation of common representation space. Theoretically, it represents the root mean square of the volatilities of the underlying during the remaining life of the option, provided the underlying follows diffusion dynamics. This variable is used for some type of scaling at the stage of defining features in the common representation space. We derive its expressions under the constant volatility as well as stochastic volatility scenarios. These expressions allow one to estimate the value of the volatility scalar. However, we retain simplicity in the implementation by assuming constant volatility while estimating the variable for investigating whether the notion of domain adaptation works even with such a crude assumption. It is also worth noting that since we do not calibrate a theoretical model for the option price prediction, the discussion of no-arbitrage calibration like \cite{Fengler2009} is irrelevant here. We also do not address the question of hedging options. A comprehensive literature survey on the statistical hedging problem can be found in \cite{RW2020}.

The rest of the paper is organized as follows. In Section \ref{sec:TheoreticalFramework}, we develop the theoretical framework of approaches based on HH and one for addressing DS. For achieving the domain adaptation, a common representation space is formulated using an approximate identity in Section \ref{sec:Domain Adaptation}. In this section, we develop an ensemble approach too. The empirical data is described in Section \ref{sec:data}, using exploratory data analysis. Subsequently, we describe all proposed supervised learning approaches in Section \ref{sec:description-learning-approach}. In Section \ref{sec:model-performance}, we report the performance of every model based on proposed approaches. Some concluding remarks are added to Section \ref{sec:conclusion}.

\section{Theoretical Framework}\label{sec:TheoreticalFramework}
\subsection{Non-parametric Settings}
Let $(\Omega, \mathcal{F}, \mathbb{P}, \{\mathcal{F}_u\}_{u\ge 0})$ be a filtered probability space. In this section we assume that $ S^{(t,s)}:=\{ S^{(t,s)}(u), u \geq t\}$ is an adapted right continuous with left limit (rcll) process that models the price of a risky asset for all future time $u\ge t$ and at present time $t (\ge 0)$ it is equal to a positive value $s$. Evidently, $S=\{ S^{(t,s)} \mid t\ge 0, s> 0\}$ forms a family of adapted processes. With a minor but usual abuse of terminology, we often call $S$ itself as a process. Let $B:=\{B(u)\}_{u\ge 0}$ denote the price process of a risk-free asset. We further assume the following.
\noindent \textbf{Model Assumptions:-}
\begin{enumerate}
\item [M1.] The market model, consisting of risky and risk-free assets with price processes $ S^{(0,s)}$ and $B$ respectively is free of arbitrage under admissible strategies for any $s>0$. 
    \item [M2.] The family of processes $S$ satisfies
\begin{align}\label{SP-condn}
  S^{(t,s)}(u) = sS^{(t,1)}(u),
\end{align}
\noindent for all $u\ge t\ge 0$ and $s> 0$ almost surely.
\item [M3.] Either $S^{(0,s)}$ is Markov or there is an additional auxiliary process $V:=\{V(u)\}_{u\ge 0}$ such that $\{(S^{(0,s)}(u), V(u))\}_{u\ge 0}$ is Markov w.r.t. the filtration $\{\mathcal{F}_u\}_{u\ge 0}$ for all $s>0$.
\end{enumerate}
It is worth noting that M1 is a standard technical assumption. On the other hand, M2 asserts affine structure, i.e. $S$ is expressed as exponential. This is valid, for example, with Geometric Brownian motion (GBM), regime-switching GBM, Merton's Jump Diffusion model, etc. Additional $V$ in M3 is not needed for GBM, but that is essential for other multi-factor models. For example, $V$ is the Markov regime in a regime-switching model, instantaneous volatility in the Heston model, etc.

\begin{definition}[Moneyness] Let $K$ be the strike price of a call option. If the present asset price is $s$, $p=K/s$ is called the present moneyness of the said option. Note that $p  = 1$ represents at-the-money, $p  > 1$ represents out-of-the-money and $p  < 1$ represents in-the-money. 
\end{definition}

\begin{thm}\label{theo2.2}
\noindent Assume that the market contains two risky assets other than a risk-free asset, whose price processes are modeled by $S_1$ and $S_2$, satisfying (M1)-(M3) with a common auxiliary process $V$. 
Let $\varphi_i$ denote the European call option price function on the $i$th risky asset. That is,  $\varphi_i (t, s_i, v; K_i, T)$ gives the present (time $t$) price of the call option on $S_i$ having exercise price and maturity 
$K_i$ and $T$ respectively with $V(t) =v$ and $S_i(t)=s_i$. Also assume that the conditional laws of $S_1^{(t,1)}(T)$ and $S_2^{(t,1)}(T)$ given $V(t)=v$ are identical under their corresponding risk-neutral measures $\mathbb{P}_1$ and $\mathbb{P}_2$, respectively. If the contracts are chosen with an identical time to maturity and equal moneyness at present time $t$ i.e., $\frac{K_1}{s_1}= \frac{K_2}{s_2} = p$ (say), we have \begin{align}\label{HHeq}
{s_1}^{-1}\varphi_1 (t, s_1, v; ps_1, T)=  {s_2}^{-1}\varphi_2 (t, s_2, v; ps_2, T). 
\end{align}
\end{thm}
The above theorem is applicable for a wide range of models. For example, when two geometric Brownian motions (GBM) with different drifts but the same diffusion coefficients are taken, their dynamics under respective risk neutral measures are identical. The requirement of common $V$ is also not a real restriction. For example, consider two regime switching GBMs $S_1$ and $S_2$, with Markov regime processes $X_1$ and $X_2$ respectively. In other words, for each $i$, $X_i$ is a finite state continuous-time Markov chain that influences otherwise constant drift and volatility coefficients of a GBM $S_i$. Then the process $V:=(X_1, X_2)$,  is such that both $(S_1, V)$ and $(S_2, V)$ are Markov. It is also worth noting that due to the normalization condition in the equality of conditional law, the theorem remains applicable even if one asset's price is more than several orders of magnitude of another. The theorem is proved below.
\begin{proof}[Proof of Theorem \ref{theo2.2}]
\noindent \noindent It is given that the market contains two risky assets, whose price processes are modeled by $S_1$ and $S_2$, satisfying (M1)-(M3) with an auxiliary process $V$. For each $i=1,2$, let $\mathbb{P}_i$ be an equivalent martingale measure (EMM) corresponding to $S_i$, i.e., $S_i/B$ is martingale under $\mathbb{P}_i$, a probability measure equivalent to $\mathbb{P}$. Using $\mathbb{P}_i$, we express a present (time $t$) fair price of a European-style call option on the $i$th risky asset having the exercise price and maturity time as $K_i$ and $T$ respectively, as 
$$\mathbb{E}_i\left(\frac{B(t)}{B(T)}\left(S_i(T)-K_i\right)^+ \mid \mathcal{F}_t \right)$$ where $\mathbb{E}_i$ is the expectation w.r.t. $\mathbb{P}_i$. Using (M3), the above is 
$$ \mathbb{E}_i\Big(\frac{B(t)}{B(T)}(S_i(T)-K_i)^+ \mid S_i(t), V(t) \Big) = \varphi_i(t,S_i(t), V(t);K_i,T), 
$$
where $\varphi_i$ denotes the call option price function on the $i$th asset with maturity $T$ and exercise price $K_i$. For each $i=1,2$, using \eqref{SP-condn}, and the above equation, and by denoting $V(t) =v$, $S_i(t)=s_i$,  we get 
\begin{align}\label{OP-condn}
   \varphi_i(t,s_i, v;K_i,T) &= \mathbb{E}_i\Big(\frac{B(t)}{B(T)}(S^{(t,s_i)}_i(T)-K_i)^+ \mid S^{(t,s_i)}_i(t)=s_i , V(t)=v\Big) \nonumber\\
   & = \mathbb{E}_i\Big(\frac{B(t)}{B(T)}s_i(S^{(t,1)}_i(T)-\frac{K_i}{s_i})^+ \mid s_iS^{(t,1)}_i(t)=s_i, V(t)=v \Big) \nonumber\\
   &  = s_i \mathbb{E}_i\Big(\frac{B(t)}{B(T)}(S^{(t,1)}_i(T)-\frac{K_i}{s_i})^+ \mid S^{(t,1)}_i(t)=1, V(t)=v\Big) \nonumber\\
   & = s_i \varphi_i\bigg(t,1,v;\frac{K_i}{s_i},T\bigg).
\end{align}
Furthermore, it is given that the conditional laws of $S_1^{(t,1)}(T)$ and $S_2^{(t,1)}(T)$ given identical 
$V$ values at time $t$ under their corresponding risk-neutral measures $\mathbb{P}_1$, and $\mathbb{P}_2$ respectively are identical. 
Hence for the contracts with identical time to maturity and equal moneyness at time $t$, i.e., $\frac{K_1}{s_1}= p=\frac{K_2}{s_2}$, we have
\begin{align*}
& \mathbb{E}_1\Big(\frac{B(t)}{B(T)}(S_1^{(t,1)}(T)-p)^+ \mid S_1^{(t,1)}(t)=1, V(t)=v\Big) \\
& = \mathbb{E}_2\Big(\frac{B(t)}{B(T)}(S_2^{(t,1)}(T)-p)^+ \mid S_2^{(t,1)}(t)=1, V(t)=v\Big)\\
\text{or, }& \varphi_1(t,1, v; p,T)= \varphi_2 (t,1,v ;p,T).
\end{align*}
Using \eqref{OP-condn}, the above equality is rewritten as 
$${s_1}^{-1}\varphi_1 (t, s_1, v; ps_1, T)=  {s_2}^{-1}\varphi_2 (t, s_2, v; ps_2, T).$$
Hence, \eqref{HHeq} holds. 
\end{proof}

\begin{rmk}[Homogeneity Hint Approach or $\mathcal{A}_{HH}$] \label{rmkHH}
The above theorem is alluded to as \emph{homogeneity hint} (see \cite{GARCIA200093} and references therein) as the scale-free terms on both sides of \eqref{HHeq} are equal. This also reflects the homogeneity of order one of the option prices in stock and strike prices. For predicting option price, the scale-free ratio of option and stock price can be considered as a response/target variable in a learning approach whereas another observed scale-free quantity, the array of order statistics of mean-adjusted historical log returns, can be taken as a predictor/feature variable. By doing so, from the input asset price data, the magnitude, temporal, and drift information are removed and the spatial distribution of return is retained.
On the other hand from the predicted target value, which stands for the ratio of option and spot prices, the predicted option price is obtained by multiplying the target with the observed stock price. Such an approach has great flexibility. It can learn from one asset class data and be applied to another, provided their log returns have similar distributions. See \cite{GRT} for more details. We call such approaches as the Homogeneity Hint Approach or $\mathcal{A}_{HH}$.
\end{rmk}

However, requiring identical distributions of log returns in two different assets is quite restrictive, even under their risk-neutral measures. Unless the features and targets are further scaled, to factor out the large disparity of their distributions, the trained model fails to perform across assets. In other words, $\mathcal{A}_{HH}$ adapts poorly to domain shifts. Indeed such an approach may find feature values of the test data from another asset unfamiliar with reference to the training data and hence fail to perform for test dataset. In view of this, we wish to formulate a common representation space, i.e., a bridge between two different risk-neutral return distributions and utilize that for connecting their corresponding option prices. To this end, we adopt parametric models of asset price dynamics.

\subsection{Parametric Settings}
In this section and onwards we only consider the processes having continuous paths, unlike rcll process in the preceding section.

\noindent The simplest parametric model for the continuous-time dynamics of asset prices is geometric Brownian motion (GBM). The GBM which is also a Markov process satisfies (M1)-(M3), with no need of auxiliary process $V$ and hence in particular\eqref{SP-condn} holds. Therefore, under GBM the theoretical call option price function $C( t,s;K,T;r,\sigma)$ satisfies
$$C( t,s;K,T;r,\sigma)=sC\bigg( t,1;\frac{K}{s},T;r,\sigma\bigg)$$
as a consequence of \eqref{OP-condn}. The function $C$ is also known as the  Black-Scholes-Merton (BSM) formula for a European call option. Here we replace $(S(t), V(t))$ by only $S(t)$, since $S$ itself is Markov. As a result, the additional variable $V$ is dropped from the option price function. The rationale behind such discrepancies in notation in other places is also evident from the context, so the discrepancies cause no ambiguity.
Next, we consider an extension of BSM model for further discussion.
Let $(\Omega, \mathcal{F}, \mathbb{F}, \mathbb{P})$ be a filtered probability space, with a Brownian motion $W={W(t)}_{t \geq 0}$ adapted to $\mathbb{F}$. Let
$S=\{S(t)\}_{t\ge 0}$ denote the unique strong solution to the following SDE
\begin{equation}\label{eq-gbm}
    d S(t) = \mu(t)S(t) dt + \sigma(t) S(t) dW(t),
\end{equation}
where $\{\mu(t)\}_{t \ge 0}$ and $\{\sigma(t)\}_{t \ge 0}$ are adapted and bounded processes, such that (M1)-(M3) hold. The solution admits the explicit representation\begin{equation}\label{eq:explicit-soln}
	S(t) = S(0) e^{\int_0^t \left(\mu(u) - \frac{\sigma(u)^2}{2}\right) du+ \int_0^t \sigma(u) dW(u)}
\end{equation}
for an arbitrary initial value $S(0)>0$. 
\begin{definition}[$\rho$-scaling of a process]
Given a positive valued process $S= \{S(t)\}_{t\ge 0}$, define $A:= \{A(t)\}_{t \ge 0}$ such that ${A(t)}^{\rho}= S(t)$ for all $t\ge 0$ and for some $\rho > 0$. 
We call the process $A$ as $\rho$-scaling of process $S$. 
\end{definition}
The $\rho$-scaling is an easy recipe for obtaining a parametric family of processes, all having different risk-neutral return distributions. A more precise version of this statement is given in the following lemma. 
\begin{lem} \label{lem1}
Let $\rho>0$. Assume that $S$ satisfies \eqref{eq-gbm} such that (M3) holds with no additional auxiliary process $V$. Then \\
(i) The $\rho$-scaling satisfy (M1)-(M3). \\(ii) The instantaneous volatility of the  $\rho$-scaling of $S$ is identical to the instantaneous volatility of $S$ divided by $\rho$.
\end{lem}
\begin{proof}
\noindent (i) 
Using the explicit solution \eqref{eq:explicit-soln}, the $\rho$-scaling satisfies
\begin{equation}\label{2.6}
    A(t) = S(t)^{\frac{1}{\rho}}= A(0) e^{
\left[\rho^{-1} \int_0^t \left(\mu(u) - \frac{\sigma(u)^2}{2}\right) du +\int_0^t \frac{\sigma(u)}{\rho}dW(u) \right]} 
\end{equation}
as $S(0)^{\frac{1}{\rho}} = A(0)$. Therefore, $A$ also satisfies \eqref{SP-condn} and (M3). Similarly, it is not hard to see that 
$A$ fulfill (M1), as an application of Girsanov's transformation and the fundamental theorem of fair pricing (See  \cite[Example 2 in Subsection 7.4]{KK00}. 

\noindent (ii) Using It\^{o}'s lemma, from \eqref{2.6}, we get 

\begin{align}\label{eq-SimpleReturn-A}
\nonumber
d A(t) =& \left( \frac{\mu(t)}{\rho} - \frac{\rho-1}{2 \rho^2}\sigma(t)^2  \right) A(t) dt + \frac{\sigma(t)}{\rho} A(t) dW(t)\\
    \textrm{i.e., } \frac{1}{A(t)} d A(t) =& \left( \frac{\mu(t)}{\rho} - \frac{\rho-1}{2 \rho^2}\sigma(t)^2  \right)  dt + \frac{\sigma(t)}{\rho} dW(t).
\end{align}
The left side is the instantaneous simple return of $A$. Thus the instantaneous volatility of $A$, i.e., the coefficient in $dW(t)$ term, is $\rho^{-1}$ times the volatility of $S$.
\end{proof}
Under the assumption on $S$, although Lemma \ref{lem1} implies that $A$ satisfy (M1)-(M3), Theorem \ref{theo2.2} is not applicable between them. In other words, an identity like \eqref{HHeq} between the option prices on $S$ and $A$ does not hold, as their returns' distributions under the respective risk-neutral measures are not identical.
However, given a pair of assets whose prices are modeled by \eqref{eq-gbm} with different deterministic parameter values, one can choose a pair of $\rho$ values so that the assets' $\rho$-scalings satisfy an identity like \eqref{HHeq}. This fact is stated in the following theorem.

\begin{thm}\label{theo2.5}
Assume that for each $i=1,2$, $S_i$ solves \eqref{eq-gbm} with $\mu =\mu_i$, $\sigma =\sigma_i$ and Brownian motion $W=W_i$. For a fixed $t\in [0,T)$, denote
\begin{align}\label{rho_cond}
\rho_i:= \left(\frac{1}{T-t}\mathbb{E}\int_t^T \sigma_i^2(u) du\right)^{\frac{1}{2}}, \quad \forall i=1,2,
\end{align} 
where $\mathbb{E}$ is the expectation with respect to $\mathbb{P}$, and denote the price function of a European call option on the $\rho_i$-scaling of $i$th asset, by  $\psi_i$. Then, if $\mu_i$ and $\sigma_i$ are deterministic processes for each $i=1,2$, \\
(i) the risk-neutral distribution of log return of the $\rho_i$-scaling of $i$th assets are identical,\\
(ii) for all $a_1>0$, $a_2>0$, $p>0$
\begin{align}\label{eqphi}
\frac{1}{a_1} \psi_1(t,a_1, pa_1, T)= \frac{1}{a_2} \psi_2(t,a_2, pa_2, T).
\end{align}
\end{thm}
\begin{proof}[Proof of Theorem \ref{theo2.5}] 
Note that $S_i$ solves for all $t> 0$
\begin{equation}
\nonumber d S_i(t) = \mu_i(t)S_i(t) dt + \sigma_i(t) S_i(t) dW_i(t),
\end{equation}
where the Brownian motions $W_1$ and $W_2$ may or may not be independent. Since $\mu_i$ and $\sigma_i$ are deterministic processes, corresponding $S_i$'s satisfy Assumption (M3) with no additional auxiliary process $V$.

Now we fix a $t>0$ and for each $i=1,2$, set $\rho_i$ arbitrarily and $ A_i(u) = S_i(u)^{\frac{1}{\rho_i}}$ for all $u\ge t$ and also denote the European call option price function on the asset having price process $A_i$ as $\psi_i$.
In view of \eqref{eq-SimpleReturn-A}, $A_i$ solves the following SDE
$$ dA_i(u) = A_i(u)\left(\frac{dB(u)}{B(u)} + \frac{\sigma_i(u)}{\rho_i} dW_i'(u)\right)$$
where, $W_i'(\cdot) := W_i(\cdot) + \int_0^\cdot \left( \frac{\mu_i(u)}{\sigma_i(u)} - \frac{\rho_i-1}{2 \rho_i} \sigma_i(u)\right) du - \int_0^\cdot\frac{\rho_i}{\sigma_i(u)} \frac{dB(u)}{B(u)}$. 
Hence 
\begin{align}\label{SDE-A}
\frac{A_i}{B}(\cdot)= \frac{A_i}{B}(t)\mathcal{E}\left(\int_t^\cdot \frac{\sigma_i(u)}{\rho_i} dW_i'(u)\right)
\end{align}
where $\mathcal{E}(X)$ denotes the Dol\'{e}ans-Dade exponential of the semi-martingale $X$.
For each $i=1,2$, let $\mathbb{P}_i'$ denote a probability measure equivalent to $\mathbb{P}$ such that $W_i'$ is a Brownian motion under $\mathbb{P}_i'$. 
This exists as a consequence of Girsanov's Theorem. Then $A_i/B$ is a martingale under $\mathbb{P}_i'$, and hence $\mathbb{P}_i'$ is an equivalent martingale measure (EMM) corresponding to $A_i$.

Next, we choose the constants $\rho_1$, and $ \rho_2$ so that under the respective risk-neutral measures $\mathbb{P}_1'$, and  $\mathbb{P}_2'$, the conditional laws of $A_1^{(t,1)}(T)$ and $A_2^{(t,1)}(T)$ are identical. That is, using \eqref{SDE-A}
\begin{align} \label{sigrho}
\int_t^T \frac{\sigma_1(u)}{\rho_1}dW_1'(u) \overset{d}{=} \int_t^T \frac{\sigma_2(u)}{\rho_2}dW_2'(u)
\end{align}
where the equality ``$\overset{d}{=}$'' is in the sense of distribution. Since the conditional distributions of $A_i(T)/A_i(t)$ under the corresponding risk neutral measures are identical, from Lemma \ref{lem1} and Theorem \ref{theo2.2} we get \eqref{eqphi}. 
Moreover, as $\sigma_i$ is bounded and deterministic, $\int_t^T \frac{\sigma_i(u)}{\rho_i} dW_i'(u)$ becomes Gaussian random variable with mean zero and variance $\rho_i^{-2}\mathbb{E}\int_t^T \sigma_i^2(u) du$ for each $i$. Hence the values of the parameters $\rho_1$ and $\rho_2$ specified by \eqref{rho_cond} are sufficient to ensure  \eqref{sigrho}. Hence the proofs of (i) and (ii) are complete.
\end{proof}

Although $\rho$ in \eqref{rho_cond} depends on the law of $\sigma$, and values of $t$, and $T$, we don't write it here to avoid clumsy notation. However, we refer to the value in \eqref{rho_cond} as \emph{volatility scalar} to avoid ambiguity. For the rest of the paper, $\rho$ and $\rho_i$ denote the volatility scalars of $S$ and $S_i$ respectively, unless otherwise mentioned.
\begin{rmk}[Estimation of the volatility scalar $\rho$]\label{Rem:rho}
From \eqref{rho_cond}, it is also evident that, if $\sigma_1$ and $\sigma_2$ are stationary, $\rho_1$, $\rho_2$ can be estimated from the past data. More precisely, $\rho$, the average of forward variance, see \eqref{rho_cond}, is typically estimated by taking the past 20 days' data and we assume that the past 20 days' historical volatility is a good estimator of future volatility, at least in average. It is worth mentioning that $\rho$ must be re-estimated as $t$ and $T$ change which is a sensible thing to do as in practice one needs to calibrate the model at every epoch (i.e., as $t$ changes and then time to maturity $T-t$ also changes). 

The selection of $\rho$, so that \eqref{rho_cond} holds becomes particularly simple if the volatility factor $\sigma$ is constant in time. Indeed, the volatility scalar $\rho$ becomes identical to $\sigma$. Even if the volatility process is neither constant nor stationary but follows a parametric diffusion equation, the volatility scalar may be estimated by estimating the parameters. We illustrate this by considering the Heston model. In that, the square of volatility follows a Cox–Ingersoll–Ross (CIR) process, i.e., for all $t\ge 0$
$$ d\sigma^2(t) = \kappa (\theta - \sigma^2(t)) dt + \xi \sigma(t) d\tilde{W}(t), 
$$
with $\sigma^2(0) >0$ where $\kappa, \theta, \xi, \sigma(0)$ are positive parameters and $\tilde{W}$ is a standard Brownian motion. A direct calculation shows that 
$$
\frac{1}{T-t}\mathbb{E}\int_t^T \sigma^2(u) du = \theta - (\theta-\sigma^2(t))\frac{1-e^{-\kappa (T-t)}}{\kappa (T-t)}.
$$
Therefore, $\rho = \sqrt{\left( \theta  - (\theta-\sigma^2(t))\frac{1-e^{-\kappa (T-t)}}{\kappa (T-t)} \right)} $, a number between the long-run mean of volatility $\sqrt{\theta}$ and the current volatility $\sigma(t)$, can be estimated using the estimator of $\theta, \sigma(t)$ and the speed parameter $\kappa$ of mean reversion.

On the other hand, for a special case of non-parametric settings, consider two stocks whose prices satisfy \eqref{eq-gbm} with identical volatility processes, i.e., $\sigma_1 \overset{d}{=} \sigma_2$. 
Then \eqref{sigrho} holds with any equal values of $\rho_1$ and $\rho_2$. 
\end{rmk}
In the aim of building a market-resilient learning approach for option price prediction, we follow the approach of common representation. First of all, we are required to select features that can help to extrapolate for successfully predicting option prices corresponding to atypical out-sample data. In this connection, Theorem \ref{theo2.5} (i) is useful. In particular, it implies that the risk-neutral distribution of log return of the $\rho$-scaling of the asset is a market-resilient feature influencing the option pricing mechanism. 
Needless to say, the learning approach should also have features coming from contract parameters and risk-free assets. This feature representation from source and target domains is indistinguishable. Next, we are required to decide on a suitable target that depends solely on features and not on factors, not included in the feature for formulating the common representation space. There is no unique or obvious way to do that. For example, a relation between $\varphi_i$ and $\psi_i$, the observable and unobserved option prices of the assets with prices $S_i$ and $A_i$ for each $i$, may be derived. Then a relation between $\varphi_1$ and $\varphi_2$, the observable option prices of the assets with prices $S_1$ and $S_2$ can be derived, as Theorem \ref{theo2.5} relates $\psi_1$ and $\psi_2$. It seems that this is hard in general. If for each $i=1,2$,
\begin{align}\label{eq:F}
\psi_i(t,a_i, pa_i, T) = F(\varphi_i(t, a_i^{\rho_i}, p a_i^{\rho_i}, T), a_i,\rho_i, p,T)
\end{align}
for some known $F$,  then from Theorem \ref{theo2.5} we have 
\begin{align}\label{eq11}
\frac{1}{s_1^{1/\rho_1}} F(\varphi_1(t, s_1, p s_1, T),  s_1^{1/\rho_1},\rho_1, p,T) = \frac{1}{s_2^{1/\rho_2}} F(\varphi_2(t, s_2, p s_2, T),  s_2^{1/\rho_2},\rho_2, p,T) =\mathcal{U} \text{ (say)}
\end{align}
provided $\rho_1$ and $\rho_2$ are chosen as in \eqref{rho_cond}. The target $\mathcal{U}$ is invariant for two different contracts with identical moneyness and time to maturity on two different assets, possibly with domain shifts. Also, $\mathcal{U}= \psi_i(t,a_i, pa_i, T)/a_i$, which depends solely on the $A_i$ features and the features related to the contract and risk-free asset. One may perhaps obtain some other targets based on a transformation involving implied volatility etc so that similar common representation is created. Assume that a learning approach $\mathcal{A}$ is trained to take centered return distribution of $S_i^{1/\rho_i}$ (i.e., the returns of $A_i$ after being subtracted by its mean) and variables and parameters $p$, $T$, $r$ as inputs and to give $\mathcal{U}$ as target. Then $\mathcal{A}$ learns a fundamental property of option pricing from the common representation space, that holds across assets having a wide range of return distributions. Further, if the map $\varphi \mapsto \mathcal{U}$ is invertible, then for test data, the predicted target may be used for computing the option price. We explain this idea by making use of an approximation formula \cite{Bharadia96} in the next section.

\section{Application of an Implied Volatility Estimate in Domain Adaptation}\label{sec:Domain Adaptation}
\subsection{A Common Representation Space} In this section, instead of an equality \eqref{eq:F}, which is hard to find, we use an approximate equality relation to derive an approximation of the equality \eqref{eq11}. To this end, we refer to \cite{Bharadia96}, which gives a simple approximation of implied volatility, the inversion of the BSM formula for near ATM options. The literature on the inversion of the BSM formula is reasonably rich. We cite \cite{Li05} for various other improvements over the simple formula in \cite{Bharadia96} which we use in this section. 
For each asset $i=1,2$, the approximation formula for $IV^{S_i, p_i, T_i}$, the implied volatility of asset $S_i$ corresponding to a call contract with moneyness and time to maturity $ p_i$, and $ T_i$ respectively, is given by 
\begin{align*}
IV^{S_i, p_i, T_i} \approx & \sqrt{\frac{2 \pi}{T_i}} \left(\frac{\varphi_i(0, s_i, p_i s_i, T_i)}{s_i(1+p_i^*)/2} - \frac{(1-p_i^*)}{(1+p_i^*)}\right), 
\end{align*}
where $p_i\approx 1$, $p_i^*= p_ie^{-rT_i}$, $s_i$ is the present price of $i$th stock. We recall that $\varphi_1$ and $\varphi_2$ are price functions of European call contracts on assets having prices $S_1$ and $S_2$ respectively which follow diffusion equations as in \eqref{eq-gbm}. Since, $\rho_i$ is defined (see Definition \ref{rho_cond}) using the instantaneous volatility values on the time interval between present and maturity time, the implied volatilities should nearly be proportional to $\rho_i$, provided all other parameters are kept constants. That is if $p_1\approx p_2$, and $T_1 \approx T_2$, we have
$$\frac {IV^{S_1, p_1, T_1}}{ \rho_1} \approx  \frac{ IV^{S_2, p_2, T_2}}{\rho_2}.
$$
However, the stylized facts and other theoretical validations of volatility smile indicate that the said approximation should worsen as $T_1$, and $p_1$ depart from $T_2$, and $p_2$. Finally, by plugging in the approximate expressions in the above approximate equality we get
\begin{align}\label{eqpsi}
\sqrt{\frac{2 \pi}{T_1}} 
 \left(\frac{\varphi_1(0, s_1, p_1 s_1, T_1)}{s_1(1+p_1^*)/2} - \frac{(1-p_1^*)}{(1+p_1^*)}\right)/\rho_1 
 \approx & \sqrt{\frac{2 \pi}{T_2}} 
 \left(\frac{\varphi_2(0, s_2, p_2 s_2, T_2)}{s_2(1+p_2^*)/2} - \frac{(1-p_2^*)}{(1+p_2^*)}\right)/\rho_2.
\end{align}
We denote the left and right sides by $\mathcal{U}(s_i,r, \sigma_i, p_i,T_i,\rho_i)$ with $i=1,2$ respectively. Also, we can approximate $\mathcal{U}(s_i,r, \sigma_i, p_i,T_i,\rho_i)= \frac {IV^{S_i, p_i, T_i}}{ \rho_i} \approx IV^{A_i, p_i, T_i} $, the implied volatility corresponding to the asset $A_i$. The last approximate equality is due to the identical values of moneyness and time to maturity and the fact that the volatility scalar value corresponding to $A$ is one.
Hence, $\mathcal{U}(s_i,r, \sigma_i, p_i,T_i,\rho_i)$ may be regarded as primarily dependent on asset $A_i$ features and the features related to contract and risk-free asset.

\subsection{Discussion on the approximation error}
Equation \eqref{eqpsi} expresses an approximate relation between these two option prices and conveys that $\mathcal{U}$ is nearly constant wrt $s, \sigma$, and the volatility scalar $\rho$. 
In particular, if prices of a pair of stocks satisfy \eqref{eq-gbm} with $\sigma_1 \overset{d}{=} \sigma_2$, then \eqref{sigrho} holds with $\rho_1=\rho_2 =\rho$, for any arbitrary value. With $\rho_1=\rho_2=1$, clearly $S_i=A_i$ and $\psi_i=\varphi_i$ for each $i=1,2$. 
Consequently \eqref{eqphi} gives $\frac{\varphi_1(0, s_1, p_1 s_1, T)}{s_1} = \frac{\varphi_2(0, s_2, p_2 s_2, T)}{s_2}$. On the other hand if $T_1=T_2=T$, $p_1=p_2=p$, Equation \eqref{eqpsi} gives
\begin{align*}
 \left(\frac{\varphi_1(0, s_1, p s_1, T)}{s_1} - \frac{(1-p^*)}{2}\right)/\rho  
\approx & \left(\frac{\varphi_2(0, s_2, p s_2, T)}{s_2} - \frac{(1-p^*)}{2}\right)/\rho\\
\textrm{i.e., } \frac{\varphi_1(0, s_1, p s_1, T)}{s_1} \approx &
\frac{\varphi_1(0, s_2, p s_2, T)}{s_2} .
\end{align*}
Thus the relation in \eqref{eqpsi} is exact, i.e., `$\approx$' can be replaced by `$=$', when the volatility scalars are identical, i.e., $\rho_1=\rho_2$ and the contract moneyness and time to maturities match. 

\noindent The approximation error in \eqref{eqpsi} associated to a pair similar contracts (i.e.,  $T_1=T_2=T$, $p_1=p_2=p$) on assets having volatility processes $\sigma_1$ and $\sigma_2$ is measured by the relative error $\sup\{\frac{|\mathcal{U}_1 - \mathcal{U}_2|}{ \max(\mathcal{U}_1, \mathcal{U}_2)} \mid s_1>0, s_2>0\}$   
where, $\mathcal{U}_i= \mathcal{U}(s_i,r, \sigma_i, p, T,\rho_i)$. We compute this error for the ATM options, i.e., moneyness parameter $p=1$ under the special case of BSM models with constant volatility coefficients. The time to maturity $T$ is taken as $0.2$ year. In Figure \ref{fig:approx}, the horizontal and vertical axes represent the volatility coefficients values of two different assets. Each volatility value ranges on $[5\%, 100\%]$. The darkness of shades show the relative error values for each ordered pair $(\sigma_1, \sigma_2)$. The maximum error is less than $4.5\%$. If none of $\sigma_1$ and  $\sigma_2$ is smaller than $9\%$, the error is less than $2\%$. This numerically validates that \eqref{eqpsi} is a reasonable approximation for practical purposes where the volatilities are not too small.
\begin{figure}[h]
\centering
\begin{subfigure}{.5\textwidth}
\centering
     \includegraphics[width=\linewidth]{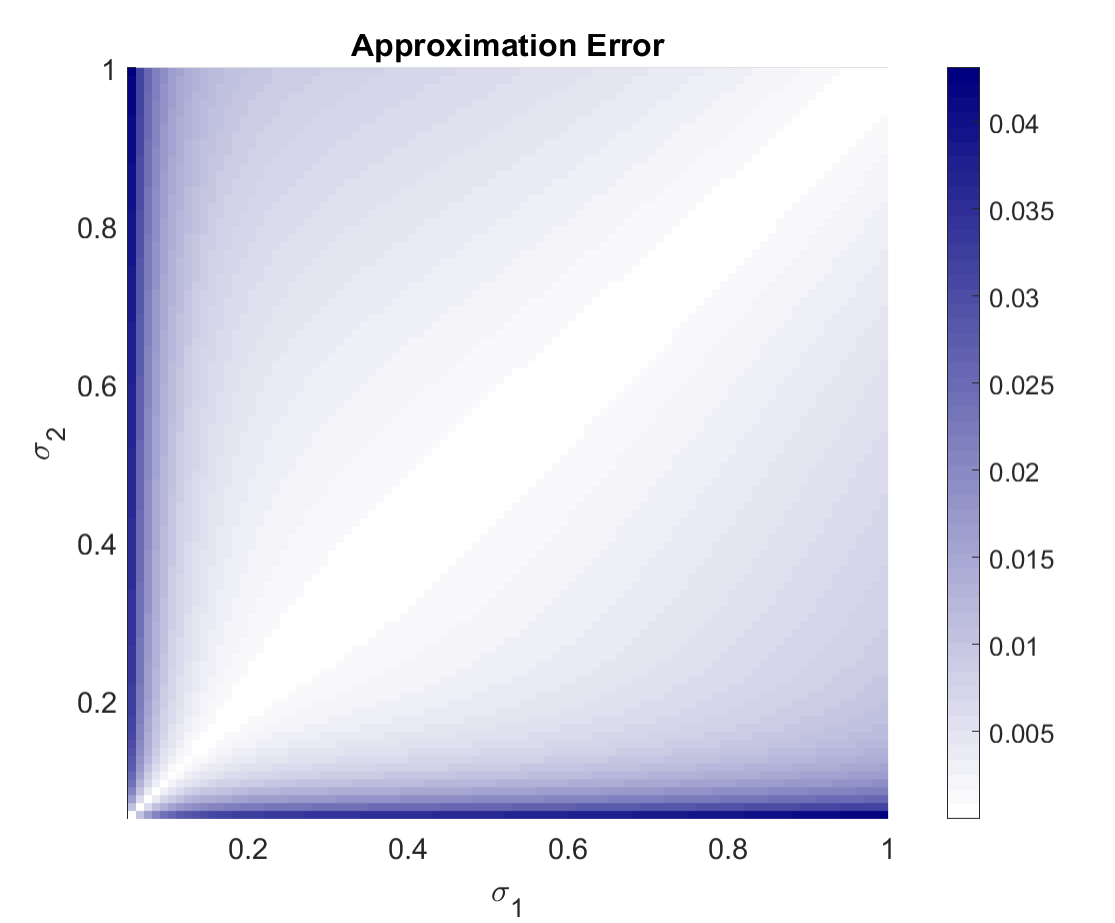}
    \caption{Grid heatmap of the approximation error for a pair of volatilities}\label{fig:approx}
\end{subfigure}
\begin{subfigure}{.48\textwidth}
\centering
     \includegraphics[width=\linewidth]{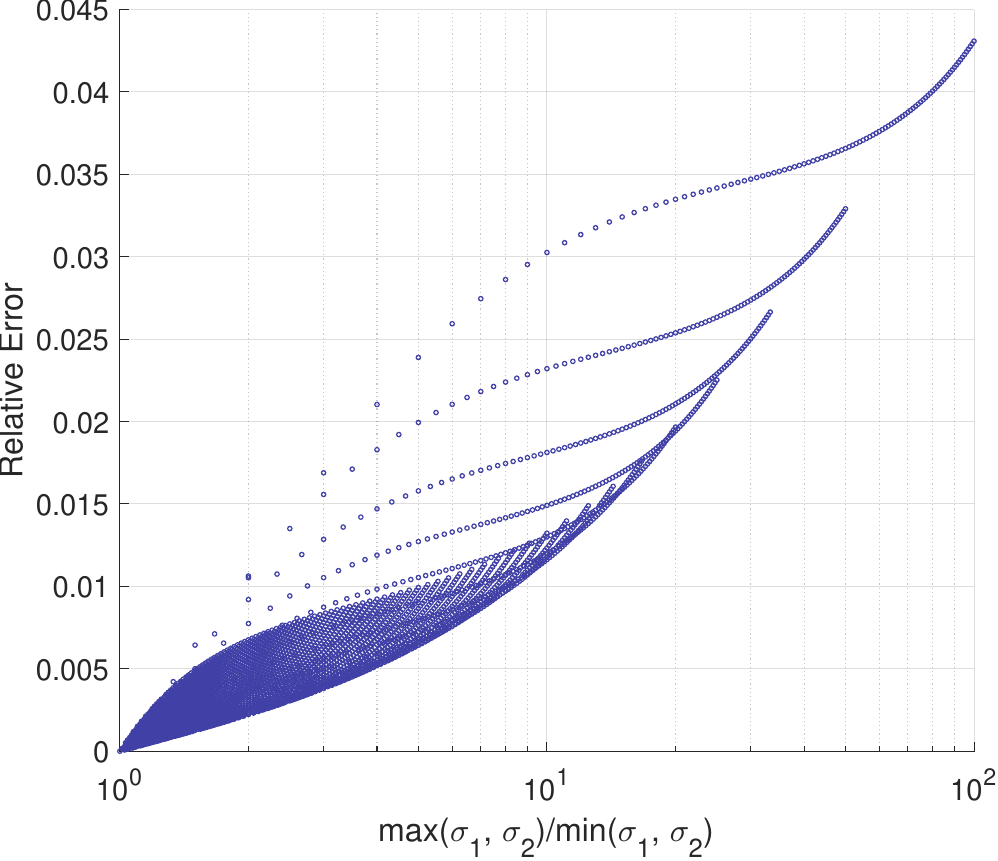}
      \caption{The approximation error vs disparity(max-min ratio) of volatilities }\label{fig:scatter}
\end{subfigure}
    \caption{Visualization of relative error in the approximation relation \eqref{eqpsi}}\label{fig1}
\end{figure}

It is evident from Figure \ref{fig:approx} that the error is minimal when $\sigma_1$ and $\sigma_2$ are identical. This can be explained by appealing to the homogeneity hint, stated under the non-parametric setting in Theorem \ref{theo2.2}. A separate visualization of error's dependence on the degree of mismatch of the volatilities is given in the scatter plot Figure \ref{fig:scatter} under the same setting. 

In this $\frac{\max(\sigma_1, \sigma_2)}{\min(\sigma_1, \sigma_2)}$, a measure of mismatch, is plotted on the horizontal axis and the corresponding relative error of the approximation \eqref{eqpsi} is plotted on the vertical axis. 
Every blue circle corresponds to a pair of volatilities $(\sigma_1, \sigma_2)$ for which the relative error is computed. The horizontal and the vertical coordinates of each circle are $\frac{\max(\sigma_1, \sigma_2)}{\min(\sigma_1, \sigma_2)}$ and  $\sup\{\frac{|\mathcal{U}_1 - \mathcal{U}_2|}{ \max(\mathcal{U}_1, \mathcal{U}_2)} \mid s_1>0, s_2>0\}$ respectively, where $\sigma_i$ is the constant volatility of $i$th asset. It shows that if the largest volatility is not larger than double the smallest, the relative error is less than $1.5\%$. On the other hand, the error is less than $3\%$ if the largest is not larger than 9 times the smallest.
\subsection{Approach $\mathcal{A}_{DS}$ for Domain Shift }\label{ssDS}
Assume that a learning approach takes risk-neutral return distribution of $\rho_1$-scaling $A_1$ of asset prices $S_1$ as a feature variable. To be more precise, the order statistics of mean subtracted log returns of $A_1$ for the past few time periods are included in the feature set. The approach should also have features from contract parameters $p$ and $T$, and risk-free assets' interest rate $r$. Also, set $\mathcal{U}(s_1,r, \sigma_1, p_1, T_1, \rho_1)$ as the target, which can be obtained from the observed option value $\varphi_1$ in the training data set. 
Consequently, for a test data set on another asset $S_2$, after feeding in the feature values to the trained model, a value of the target $\mathcal{U}(s_2,r, \sigma_2, p_2, T_2,\rho_2)$ is predicted. Then the option price can be revived using the following formula 
\begin{align}\label{eq:Output_DS}
\varphi_2(t, s_2, p_2 s_2, T_2) = s_2\left( \rho_2 \frac{1+p_2^*}{2} \sqrt{\frac{T_2}{2\pi}}\mathcal{U}(s_2,r, \sigma_2, p_2,T_2,\rho_2)  + \frac{1-p_2^*}{2}\right), 
\end{align}
irrespective of its volatility as long as $\rho_2$ is the volatility scalar for $S_2$. Therefore, the above-mentioned learning approach is better adapted for the domain shift. We denote this as $\mathcal{A}_{DS}$.
Of course, it is possible to develop a different learning model following this idea by choosing an approximation relation more accurate than \eqref{eqpsi} for possibly better performance. So, as far $\mathcal{A}_{DS}$ is concerned, we can imagine that although this simple approximation facilitates predicting option prices reasonably well for test data sets with domain shifts, would fail to outperform $A_{HH}$ when domain shift is absent. This raises the need of an ensemble model.

\subsection{Ensemble Model} \label{ssAE} The difference between the present and the past mean volatilities, normalized by the second one, may be considered as a crude measure of the domain shift. We call this measure the domain shift quotient or simply DSQ. To be more precise, if $\sigma_0$ and $\sigma_i$ are the average historical volatility of the full training data set and the current historical volatility associated with $i$th test data respectively, the formula for the DSQ for the specific test data is $\frac{|\sigma_i -\sigma_0|}{\sigma_0}$.
In this subsection, we propose an ensemble model where the ratio of weights given to the predictions from $\mathcal{A}_{DS}$ and $\mathcal{A}_{HH}$ models is $ \lambda_1 \left(\frac{|\sigma_i -\sigma_0|}{\sigma_0}\right)^{\lambda_2}$, a power function of DSQ. We denote the resulting approach as $\mathcal{A}_{E(\lambda)}$ where $\lambda =(\lambda_1, \lambda_2)$. 
If $P_{DS}(i)$ and $P_{HH}(i)$ are the predicted prices of $i$-th option from the models with approaches $\mathcal{A}_{DS}$ and $\mathcal{A}_{HH}$ respectively trained with same data, then $P_{E}(i)$, the predicted price from an ensemble approach $\mathcal{A}_{E(\lambda)}$, is defined as 

\begin{equation}\label{def:ensemble}
    P_{E(\lambda)}(i) =  \frac{1}{1 + \lambda_1 \left(\frac{|\sigma_i -\sigma_0|}{\sigma_0}\right)^{\lambda_2}} P_{HH}(i)+   \frac{ \lambda_1 \left(\frac{|\sigma_i -\sigma_0|}{\sigma_0}\right)^{\lambda_2} }{1 + \lambda_1 \left(\frac{|\sigma_i -\sigma_0|}{\sigma_0}\right)^{\lambda_2}} P_{DS}(i).    
\end{equation}
As a result, for a fixed $\lambda$, the weight to $P_{HH}$ diminishes as the atypicality rises in the test data. The more atypical the return data, i.e., the larger the relative difference between the test and the training volatilities, the higher weight is assigned to the prediction $P_{DS}$ in the ensemble approach $\mathcal{A}_{E(\lambda)}$. The parameter $\lambda$ is estimated using the least square method from a sample of test data. The ensemble approach using the estimated $\lambda$ is simply denoted by $\mathcal{A}_{E}$. The performance of the ensemble model may be evaluated for in-sample and out-of-sample data.

\section{Data}\label{sec:data}
According to the Futures Industry Association (FIA), 84\% of global equity options are now traded on Indian exchanges, up from just 15\% a decade ago.\footnote{Article by Bennett Voyles \href{https://www.fia.org/marketvoice/articles/premium-turnover-indian-options-hits-150-billion}{https://www.fia.org/marketvoice/articles/premium-turnover-indian-options-hits-150-billion}} In the first quarter of 2024 alone, over 34.8 million equity index options were traded on India’s two main financial derivatives exchanges—the National Stock Exchange of India (NSE) and the Bombay Stock Exchange (BSE) India. This was double the volume from Q1 2023 and accounted for 73\% of all futures and options traded globally.
Despite this high trading volume, the market's value remains relatively small compared to North American and European markets due to lower notional values and premiums per contract. For example, in March 2024, NSE's total options premiums were \$150 billion, about a quarter of the \$598 billion seen across all US options exchanges. Nonetheless, the Indian market's rapid growth and substantial trading volume makes it a crucial area for developing and applying data-driven option trading strategies.

For our study, we analyze daily contract price data for European call options on NIFTY50\footnote{NIFTY 50, abbreviated as NIFTY50, serves as a benchmark index for the Indian stock market. It reflects the weighted average performance of 50 of the largest companies listed on the National Stock Exchange of India.} and BANKNIFTY\footnote{The NIFTY Bank index, referred to as BANKNIFTY, consists of the most liquid and highly capitalized Indian banking stocks.}.
The abundance of data allows a learning model to better capture and depict market behavior and trading patterns. On the other hand, due to the presence of many traders' participation, the traded prices are free from market inefficiency. The dataset, spanning January 1, 2015, to April 30, 2020, is obtained from the NSE’s historical archives\footnote{Link to access the data -\href{https://www1.nseindia.com/products/content/derivatives/equities/historical_fo.htm}{https://www1.nseindia.com/products/content/derivatives/equities/historical\_fo.htm}}. 
The data-driven supervised learning approaches for pricing European-style call options, outlined in \cite{GRT} reportedly performed worse in abrupt highly volatile market conditions. Our objective is to introduce a novel market-resilient approach for improving the prediction for out-sample data having atypical asset returns. To this end, we incorporate testing of the trained model with data from the COVID-19 lockdown period, when the financial market crashed, i.e., from January 2020 to April 2020.

Each data point includes the following information, (i) the daily closing prices of the underlying asset and a particular European call option, (ii) moneyness and the time to maturity of that call contract, (iii) daily 19 log returns of the underlying asset from the 20 most recent daily data, (iv) daily close price of the same call contract on the current and earlier day and v) prevailing risk-free interest rate. Moreover, we filter out all data points for which the same option contract did not exist the previous day.
Our choice of a 20-day window aligns approximately with a calendar month, accounting for holidays. Since we compute the daily 19 log returns of the underlying asset, to have these log returns from January 1, 2015, to April 30, 2020, we consider the daily data of the underlying assets (specifically, NIFTY50 and BANKNIFTY indices) from October 1, 2014, to April 30, 2020. 

We retain datapoints only for option contracts that are near at-the-money (ATM), specifically those where the ratio of the strike price to the spot price is between 96\% and 104\%. We refer to these as near-ATM option contracts. It is observed that numerous near-ATM contracts are traded daily, with varying times to maturity. Since, contracts with large or minimal times to maturity have significantly low trading volumes, we focus on studying contracts with time-to-maturity values between 3 and 45 days.

\noindent \textbf{Train/Test Split:} To develop a predictive model using XGBoost supervised learning, we partitioned the dataset into two segments, allocating roughly $80\%$ for training and $20\%$ for evaluating the trained models. The split is such that the training set includes data dated before the oldest data of the test set. This is for avoiding information leakage as explained in \cite{WR2022}.
Specifically, the training dataset spans $56$ months, covering the period from January 01, 2015, to August 31, 2019.

\noindent Subsequently, the test dataset is subdivided into two parts to facilitate comprehensive evaluation: 
\begin{enumerate}
    \item \textbf{Typical}: The ``typical" test dataset encompasses data from September 01, 2019, to December 31, 2019. This period is characterized by typical market dynamics, enabling model evaluation under market behaviour similar to that observed during the training phase.
    \item \textbf{Atypical}:  The ``atypical" test dataset consists of data from January 01, 2020, to April 30, 2020. This period includes the COVID-19 Indian financial market crash, presenting a unique testing scenario to assess model robustness under an extremely volatile market. The dynamics during this period are significantly different from those observed during the training phase, thereby offering valuable insights into the model's performance in unprecedented market scenarios.
\end{enumerate}
Table \ref{tab:test_train_split} provides a breakdown of the number of data points at each stage of the model building and evaluation process, offering insight into the dataset's composition and distribution across training and test phases.

\begin{table}[ht]
\caption{Train/Test split of European style call option contract dataset sizes}
\label{tab:test_train_split}
\begin{tabular}{@{}|c|c|c|c|@{}}
\toprule
\textbf{\textbf{Dataset}}  &  \textbf{Time Period}   & \textbf{N50} & \textbf{BNF} \\ \midrule
\textit{Raw}    &  01/01/2015 - 30/04/2020 &  1611549  &     576302          \\ \midrule
\textit{Filtered}    &  01/01/2015 - 30/04/2020 &  27372  &          37284     \\ \midrule
\textit{Train} & 01/01/2015 - 31/08/2019  &     19182  &     26348           \\ \midrule
\textit{Typical Test}  &  01/09/2019 - 31/12/2019 &    3692   &     3709           \\ \midrule
\textit{Atypical Test}   &  01/01/2020 - 30/04/2020 &  3047      &  2774              \\ \bottomrule
\end{tabular}
\end{table}

It is important to highlight that the sum of data points in the training, typical test, and atypical test datasets do not match the number of data points in the filtered dataset. This discrepancy occurs because, during the feature generation step in our approach, we remove the rows of data associated with option contracts that do not have a previous day option closing price, which is a feature in our learning approach. 

\begin{figure}[!h]
    \centering
    \includegraphics[width = 0.55\textwidth]{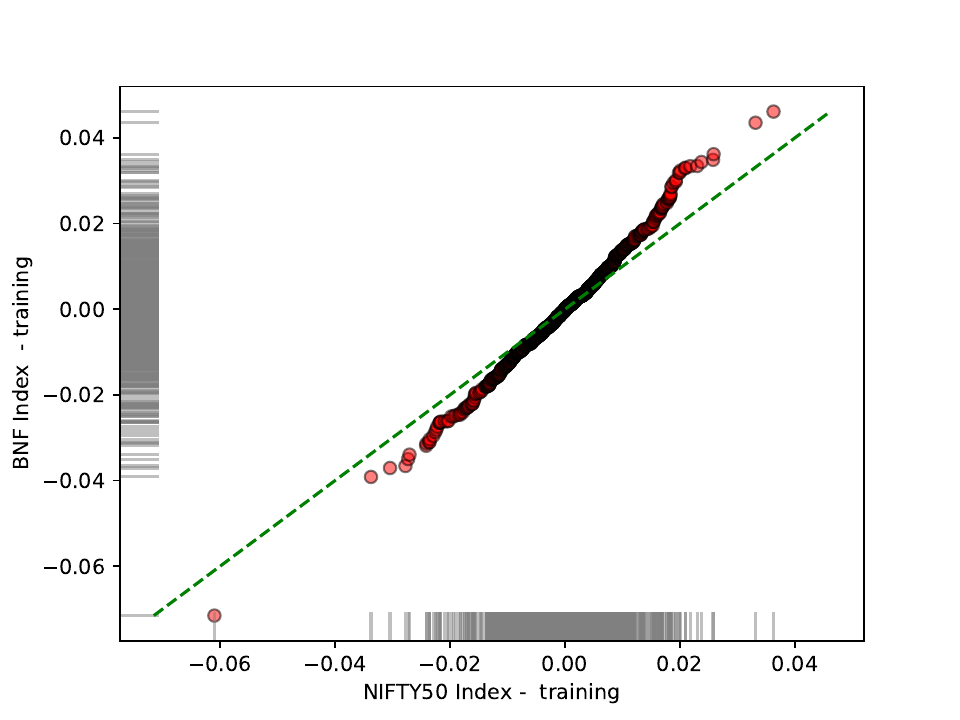}
    \caption{Q-Q plot for the log return distributions of the ``Close" prices of the BANKNIFTY and NIFTY50 indices for training dataset from January 01, 2015 to August 31, 2019}
    \label{fig:QQ_plot_training}
\end{figure}

\begin{figure}[h]
\centering
\begin{subfigure}{.49\textwidth}
\centering
     \includegraphics[width=\linewidth]{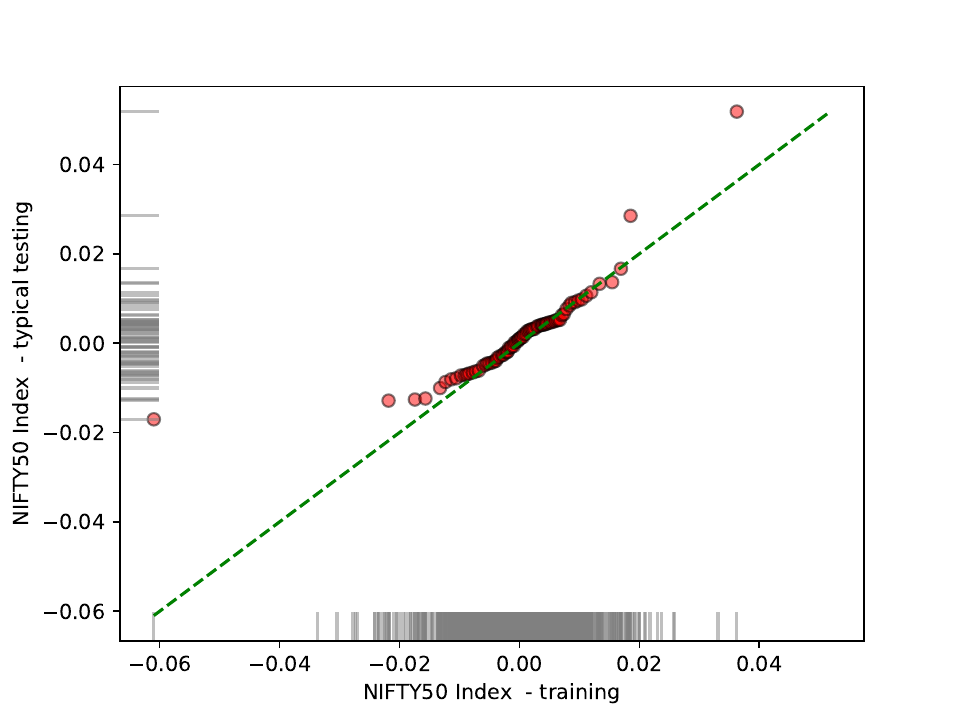}
    \caption{Train vs Typical test data}\label{fig:QQ_plot_train_typical_test}
\end{subfigure}
\begin{subfigure}{.49\textwidth}
\centering
     \includegraphics[width=\linewidth]{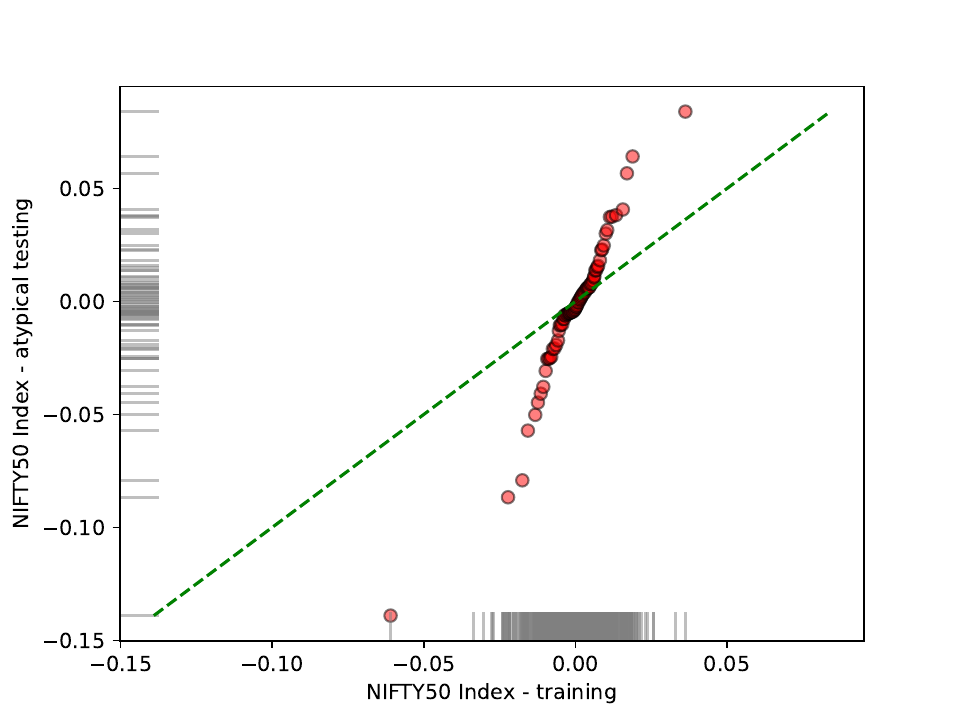}
      \caption{Train vs Atypical test data}\label{fig:QQ_plot_train_atypical_test}
\end{subfigure}
    \caption{Q-Q plot for the log return distributions of the “Close” prices of  NIFTY50 in training and testing data}\label{fig3}
\end{figure}

\begin{figure}[h]
\centering
\begin{subfigure}{.49\textwidth}
\centering
   \includegraphics[width=\linewidth]{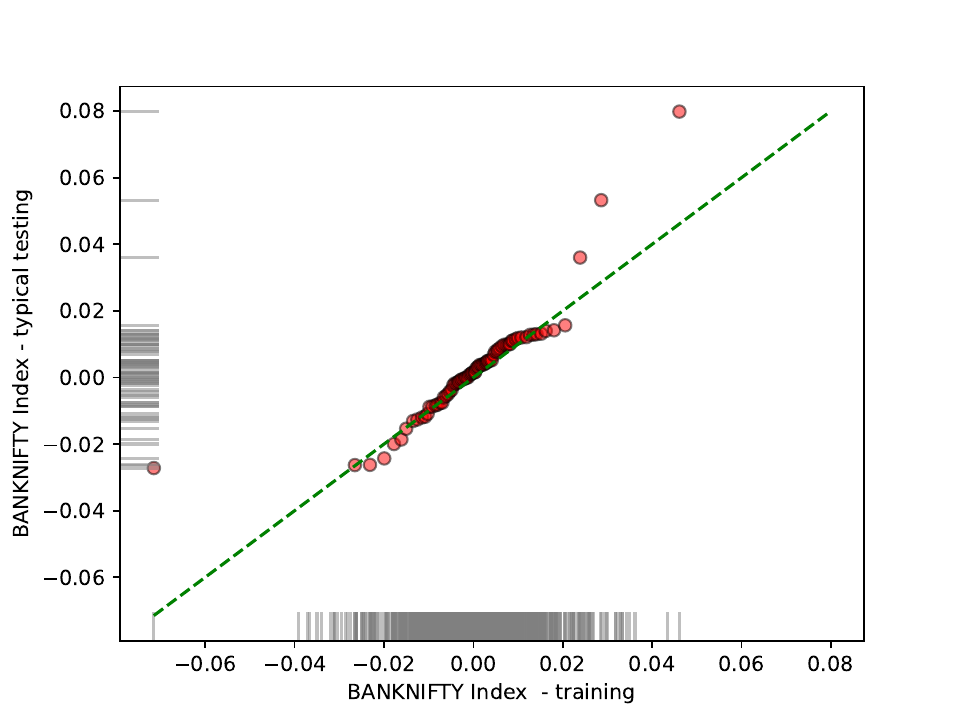}
    \caption{Train vs Typical test data}\label{fig:QQ_plot_train_typical_test_bnf}
\end{subfigure}
\begin{subfigure}{.49\textwidth}
\centering
\includegraphics[width=\linewidth]{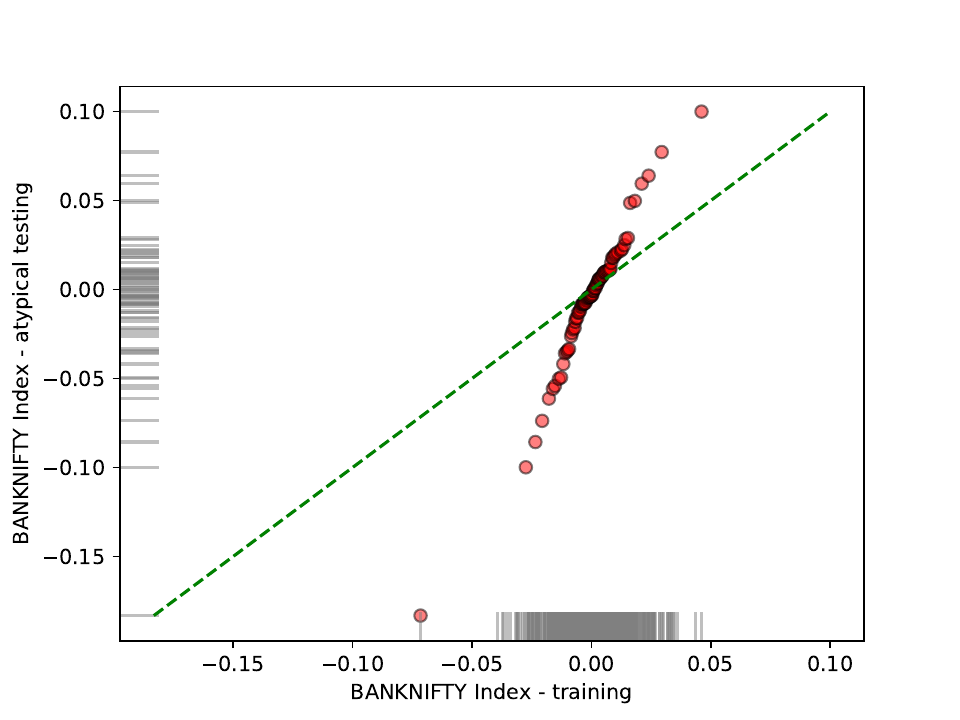}
      \caption{Train vs Atypical test data}\label{fig:QQ_plot_train_atypical_test_bnf}
\end{subfigure}
    \caption{Q-Q plot for the log return distributions of the “Close” prices of  BANKNIFTY in training and testing data}\label{fig4}
\end{figure}
A brief exploratory data analysis is presented below. The Q-Q plot in Figure  \ref{fig:QQ_plot_training} shows that the log returns of daily closing values of NIFTY50 and BANKNIFTY indices do not possess identical distribution. Nevertheless they are not significantly different during the period that constitutes the training dataset. The Q-Q plot is linear and has a slightly different slope than the identity line. Next, we compare the statistical behavior of log-returns of daily closing values of NIFTY50 index during different time intervals. In Figure \ref{fig:QQ_plot_train_typical_test} the said comparison between the data from training window and the typical test window is presented. On the other hand, Figure \ref{fig:QQ_plot_train_atypical_test} presents a comparison between the data from training and the atypical test windows. In \ref{fig:QQ_plot_train_typical_test} the Q-Q plot is aligned with the identity line but in \ref{fig:QQ_plot_train_atypical_test} the slop of the Q-Q plot is visibly far from the identity line. Similar comparisons are reported in Figure \ref{fig4} where we compare the distributions of log-returns of daily closing values of BANKNIFTY index during different time intervals.
In \ref{fig:QQ_plot_train_typical_test_bnf}, except few outliers, the return distribution of the typical test dataset matches well with that of the training dataset. The points in the Q-Q plot approximately lie on the identity line. However, when the atypical test data is compared with the training data in Figure \ref{fig:QQ_plot_train_atypical_test_bnf}, the distributions of log-returns of the closing values appear to differ drastically. Although the Q-Q plot is linear in this case, the slop vastly mismatches with that of the identity line. These observations justifies the use of terminology ``typical'' and ``atypical'' for two non-overlapping test datasets.

\section{Description of Supervised Learning Approach}\label{sec:description-learning-approach}
In this section, we outline the steps involved in the supervised learning approach utilized in this study. 

\noindent \textbf{Step 1: \underline{Data Cleaning and Filtering}}
In this essential step, we prepare the gathered raw data for European call option contracts to apply the ML algorithms. Initially, we eliminate rows containing blank `-' entries in the `Underlying Value' or `Strike Price' columns. Additionally, any rows containing values of `0' in the `Open' or `Close' columns are removed.

To further refine the selection, we follow the modularity approach of \cite{GRT} and focus on option contracts that closely resemble at-the-money (ATM) contracts. This involves filtering contracts based on a criterion where the absolute quantity $\big\vert 1- \frac{S}{K} \big\vert$ is less than or equal to the pre-decided value of 0.04 as considered in \cite{GRT}, where $S$ and $K$ respectively represent the spot price of underlying and strike price of the option contract. However, it is important to note that contracts with significantly large or small time to maturities exhibit notably diminished trading volumes. Therefore, to conduct a more focused analysis, we opt to exclusively investigate contracts with time-to-maturity (`TTM') values falling within the range of not more than 45 days and not less than three days.

\noindent \textbf{Step 2: \underline{Feature Generation}}
Let $\Delta $ denote the time granularity in year unit and $\bar S(t):= S(\Delta*t)$ the price of the risky asset at present time $\Delta\times t$, where $t$ is an integer. The discrete time series $\{\bar S(t) \mid t=0,1,\ldots\}$ gives daily data if $\Delta = 1/255$, by imagining 255 number of trading days in a year. 
Let $n$ be a fixed natural number. For each $i\le n$, the feature variable $F_i:=R_{(i)}$, where, $\{ R_{(i)} : i =1, \ldots, n\}$  is order statistics of $R_i:= \ln(\bar S({t-n+i})/ \bar S({t-n-1+i})) -\frac{1}{n}\ln(\bar S({t})/\bar S({t-n}))$, the centered daily log return. The order statistics $R_{(i)}$ are determined by arranging $R_{i}$ in ascending order for each sample. To clarify, the $i^{\text{th}}$ order statistic, denoted by $x_{(i)}$, of a sample comprising various real values $(x_1, x_2, x_3, \ldots, x_n)$ represents the $i^{\text{th}}$-smallest value, such that $x_{(i)}=x_{\beta(i)}$ where $\beta$ is a permutation on  $\{1, \ldots, n\}$, and $x_{(1)} \le x_{(2)} \le \cdots \le x_{(n)}$ is satisfied. We set $n=19$ for sampling past daily data of nearly one month. At this stage, we compile a set of 23 features to be utilized in supervised learning models: \begin{itemize}
        \item The set comprises $19$ order statistics of centered log return, denoted by  $F_1$ through $F_{19}$, derived from index (or underlying asset) data. 
        \item  The time-to-maturity (`TTM') of the option contract under study, measured in days.
        \item Moneyness of the option contracts is calculated as the ratio $\frac{K}{S}$, representing the ratio between the strike price $K$ of the option contract and the underlying value or Spot Price $S$ being examined. Its reciprocal is used as a feature.
        \item  Normalized previous option price, computed as $100 \times \frac{C_{t-1}}{S_{t-1}}$, where $C_{t-1}$ denotes the previously reported close price of the option contract under consideration and $S_{t-1}$ denotes the spot price at time $t-1$. If the option contract was not traded on the previous day (i.e., $C_{t-1}$ data is unavailable), we exclude this option from consideration. 
        \item  The ideal bank's interest rate is approximated by the three-month sovereign bond yield rates in the literature and added as a feature. This serves as an estimation for risk-free interest rates.
    \end{itemize}

\noindent \textbf{Step 3: \underline{Target Variable}}
We use regression to implement the Homogeneity Hint Approach $\mathcal{A}_{HH}$, as described in Remark \ref{rmkHH}. According to this approach (also see \cite{GRT}) `$\mathcal{A}_{HH}$ target' is the `normalized price' and is given by
$$ \mathcal{A}_{HH}  \textrm{ target} = 100 \times \frac{C}{S},
$$
a scale-free target variable.
For the regression implementation of the $\mathcal{A}_{DS}$ approach, as presented in Subsection \ref{ssDS}, the scale-free target variable referred to as the `$\mathcal{A}_{DS}$ target' is defined by \eqref{eqpsi}, i.e.,  
    \begin{align*}
 \mathcal{A}_{DS} \textrm{ target} = \frac{1}{\rho} \sqrt{\frac{2 \pi}{T}} \left(\frac{C}{S (1+p^*)/2} - \frac{(1-p^*)}{(1+p^*)}\right).
    \end{align*}
    where $C,S$, and $T$ denote the call option price, underlying value, and time-to-maturity in years, respectively.   Here, $p=\frac{K}{S}$ is moneyness,  $p^* =\frac{Ke^{-rT}}{S} = pe^{-rT}$ is discounted moneyness, where $K$ is the strike price of the call option contract and $r$ is the interest rate of an ideal bank estimated to `Yield03' value divided by $100$. The volatility scalar $\rho$ is estimated using historical volatility, the annualized (assuming 255 trading days in a year) standard deviation computed from the set of feature variables $F_1$ to $F_{19}$, derived earlier. 

\noindent \textbf{Step 4: \underline{Ensemble Model}} At this step we implement the ensemble approach  \eqref{def:ensemble} formulated in Subsection \ref{ssAE}. Since the total availability of atypical data is limited, for statistically meaningful performance analysis, all atypical data is kept for testing. However, for effective tuning of $\lambda$, the training set is required to have both typical and atypical data. Anyways, in view of the primary goal of the investigation, we give up tuning the parameters $\lambda_1$ and $\lambda_2$, and set a canonical value $\lambda_1=\lambda_2 =1$. That is the ratio of weights given to the predictions
from $\mathcal{A}_{DS}$ and $\mathcal{A}_{HH}$ is identical to DSQ.
As a result, we have
$$ P_{E(1,1)} = \frac{1}{1 + DSQ} P_{HH}+   \frac{ DSQ}{1 + DSQ} P_{DS}.$$
Although due to this simplicity, the performance may not be superior, however, it is enough for us to demonstrate the promise of an ensemble model, which is the primary goal of this study.

\noindent \textbf{Step 5: \underline{Testing and Error Metric}}
After the training phase, it is imperative to evaluate our model's performance using unseen test data, ensuring a comprehensive understanding of prediction quality. This study employs Root Mean Square Error (RMSE) to assess regression model predictions. Assume that the total number of data points in the test set is $N$ and $e_i$ is the difference between the target value and the predicted target value of a model $\mathcal{M}$ at the $i$th entry. Then $e_i$ denotes the prediction error for the $i$th test data for each $i\le N$. Finally, the RMSE of the prediction of model $\mathcal{M}$ for the test data set is defined as
\begin{equation}\label{def:RMSE}
    \textrm{RMSE} =  \left( \frac{1}{N} \sum_{i=1}^N (e_i)^2 \right)^{\frac{1}{2}}.
\end{equation}
It is crucial to emphasize that the evaluation specifically centers on gauging the regression model's performance with respect to the scale-free target variable denoted as `normalized price'. In the context of `$\mathcal{A}_{HH}$ target' scenarios, RMSE formula are applied directly to the test variable `normalized price'. However, for scenarios involving `$\mathcal{A}_{DS}$ target', we utilize \eqref{eq:Output_DS} to compute the predicted `normalized price' and subsequently evaluate the regression model's performance.

\noindent \textbf{Step 6:  \underline{Setting Numerical Values of Hyperparameters}} 
We ran the codes in Python version 3.10.4. XGBoost, a well-studied supervised ML algorithm developed by Chen and Guestrin \cite{CG16}, was utilized in this research. Table \ref{tab:hyperparams_xgb} presents the hyperparameter values used to train the XGBoost (Extreme Gradient Boosting) linear regression model described in this manuscript. While XGBoost offers numerous hyperparameters, we focused on the following key ones. For detailed information on each hyperparameter used, please refer to the XGBoost Documentation\footnote{There are many resources available to read about XGBoost parameters; for example, one can refer to \url{https://xgboost.readthedocs.io/en/stable/parameter.html}}. 
\begin{table}[!h]
\centering
\caption{Hyperparameter values for XGBoost}
\label{tab:hyperparams_xgb}
\begin{tabular}{@{}lr@{}}
\toprule
\textbf{Name} & \textbf{Value Set} \\ \toprule
n\_jobs
        & 4                     \\
objective                       
        & reg:squarederror            \\
n\_estimators                                                           & 500                \\
max\_depth                                                              & 6                  \\
learning\_rate                                                          & 0.01                \\ 
min\_child\_weight
        & 3                    \\ 
colsample\_bytree
        & 1.0                     \\ 
subsample
        & 0.6                     \\         \bottomrule
\end{tabular}
\end{table}

We conducted a cross-validated grid search using the scikit-learn GridSearchCV routine to determine the optimal values for these hyperparameters. However, we observed that the model performance did not vary significantly with changes in the hyperparameters.

\section{Model Performance}\label{sec:model-performance}
Our study involves analyzing and comparing the performance of three approaches, namely, $\mathcal{A}_{HH}$, $\mathcal{A}_{DS}$, and $\mathcal{A}_{E}$ for option price predictions.  While the first two are supervised learning models, the last one is an ensemble model, whose ensembling parameter is subject to calibration using the sample. The first approach $\mathcal{A}_{HH}$ appeared in \cite{GRT} whereas the other two are proposed in this article. 
We have collected historical option price data on two different indices N50 and BNF. Based on the data from each symbol we train three models using above mentioned three approaches. Moreover, we perform multi-source training using historical data from both symbols to get an additional combined model using each of the three approaches. Hence we obtain a total of 9 models. For each symbol of the training set, the ensembling parameter of the ensembling model is estimated from the training and test data of that symbol. We test each of these 9 models on  all four disjoint test data samples.

\subsection{Performance comparison using RMSE} The model performances are presented in Table \ref{tab:XGB performance}. In that, the models are classified into 3 classes based on the training data sets. The symbol/symbols of each training data are listed in the first column. The approaches used in the model are mentioned in the second column. The estimated values of ensembling parameters are mentioned next to each $\mathcal{A}_{E}$ approach in the second column.
The option price prediction using historical volatility and the Black-Scholes-Merton formula is set as the benchmark for comparison.
We compare the performance of all these models and the benchmark using two types of test data coming from each of the two symbols N50 and BNF.  
The performance, i.e., the RMSE of predictions, of each model is reported for all four disjoint test data samples. Given a training symbol and test data, the best-performed approach is indicated by highlighting the least RMSE value in orange. Given a test data, the least RMSE across all nine trained models is indicated with a star $(*)$ sign.

\begin{table}[ht]
\caption{Performance of XGBoost regression for $\mathcal{A}_{HH}$, $\mathcal{A}_{DS}$, and Ensemble Approaches are assessed. RMSE values for models with training (01/01/2015 - 31/08/2019), typical testing (01/09/2019 - 31/12/2019), and atypical testing (01/01/2020 - 30/04/2020) periods are reported.}
\begin{center}
\begin{tabular}{|c|c|c|c|c|c|}
\toprule
\multicolumn{6}{|c|}{\textbf{RMSEs of option price prediction models using XGBoost regression}} \\ \midrule
\textbf{Symbol(s) of}  &  \textbf{Approach} & \multicolumn{2}{|c|}{\textbf{Test - N50}}   & \multicolumn{2}{|c|}{\textbf{Test - BNF}}
\\ \cline{3-6} 
\textbf{Training data}  &   & \textbf{Typical}   & \textbf{Atypical} & \textbf{Typical} & \textbf{Atypical} \\ \midrule
    \hline
-    & Benchmark - BSM & 0.678 & 1.679 & 0.966 & 1.533  \\ \hline
    \multirow{3}{*}{\textbf{N50}} & $\mathcal{A}_{HH}$ & 0.286  & 1.339 & 0.404 & 1.313 \\
     & $\mathcal{A}_{DS}$ & 0.303 & 0.663  & 0.427 & 0.616 \\
     & $\mathcal{A}_{E(1,1)}$ & \colorbox{orange}{\textbf{0.218}}  & \colorbox{orange}{\textbf{0.554*}} & \colorbox{orange}{\textbf{0.288}} & \colorbox{orange}{\textbf{0.515}} \\
    \hline
    \multirow{3}{*}{\textbf{BNF}} & $\mathcal{A}_{HH}$ & 0.212 & 1.163 & 0.265 & 1.119 \\
     & $\mathcal{A}_{DS}$ & 0.247 & 0.743 & 0.299 & 0.641 \\
     & $\mathcal{A}_{E(1,1)}$ & \colorbox{orange}{\textbf{0.184*}} & \colorbox{orange}{\textbf{0.641}} & \colorbox{orange}{\textbf{0.223*}} & \colorbox{orange}{\textbf{0.538}} \\
    \hline
    \multirow{3}{*}{\textbf{N50+BNF}} & $\mathcal{A}_{HH}$ & 0.259 & 1.155 & 0.282 & 1.185 \\
     & $\mathcal{A}_{DS}$ & 0.247 & 0.648 & 0.312 & 0.585 \\
     & $\mathcal{A}_{E(1,1)}$ & \colorbox{orange}{\textbf{0.195}} & \colorbox{orange}{\textbf{0.574}} & \colorbox{orange}{\textbf{0.238}} & \colorbox{orange}{\textbf{0.502*}} \\
    \hline
\end{tabular}
\end{center}
\label{tab:XGB performance}
\end{table}

At the outset, we notice that for each test data, the column-wise reported RMSE values for all nine training models are smaller than the RMSE of the benchmark prediction. In fact, the lowest RMSE turns out to be smaller than $50\%$ of the benchmark RMSE for each of the four test data sets.
It is observed that on every atypical test data, the $\mathcal{A}_{DS}$ approach outperforms $\mathcal{A}_{HH}$ approach, no matter what the training data set is. This validates the usefulness of the domain shift approach presented in this paper. On the other hand, we also observe that $\mathcal{A}_{HH}$ outperforms $\mathcal{A}_{DS}$ on the typical test data of each symbol when the training data is from the same symbol. This empirically validates the theoretical conclusion of Theorem \ref{theo2.2}. We explain this below. We recall that $\mathcal{A}_{HH}$ is founded on perfect equalities, whereas $\mathcal{A}_{DS}$ is based on some approximations. So, in the special circumstances where Theorem \ref{theo2.2} is applicable, $\mathcal{A}_{HH}$ gives a better result than $\mathcal{A}_{DS}$. On the other hand, for typical test data of the same symbol of training data, as anticipated, the assumption of similarity between the return distributions of training and test is valid. 
Again, the ensemble approach outperforms both $\mathcal{A}_{HH}$ and $\mathcal{A}_{DS}$ on each of the test datasets. It is interesting to note that each of the 6 single-trained models performs similar or better on typical test data of N50 than those of BNF. On the contrary, each of those 6 models performs similar or better on atypical test data of BNF than those of N50. Such consistent performance of approaches on typical or atypical test data irrespective of the source of training data indicates the importance of consideration of domain shifts in the model building. It also indicates that while $\mathcal{A}_{HH}$ is insensitive to the domain shift, $\mathcal{A}_{DS}$ incorporates the domain adaptability reasonably well. So far prediction on BNF atypical test data is concerned, RMSE value of the ensemble model trained with N50 data is lower than that of the ensemble model trained with BNF data. This is another instance of the success of common representation space and domain adaptation.

\subsection{Performance of multi-source training model} Now, we contrast the performance of multi-source or combined training with single-source training models. We start with the following simple observation from Table \ref{tab:XGB performance}. We consider the performance of models trained on data of one symbol and tested on typical/atypical data of another symbol with the approach $\mathcal{A}_{HH}$ or $\mathcal{A}_{DS}$. RMSEs of eight such cross-symbol experiments are reported in Table \ref{tab:XGB performance}. We compare those with eight other corresponding RMSEs for the multi-source models. In seven out of eight cases, the RMSE of multi-source training is lower than those of cross-symbol experiments. 
We also observe the following. Although $\mathcal{A}_{HH}$ approach performs poorer than $\mathcal{A}_{DS}$ on atypical test data, the performance of $\mathcal{A}_{HH}$ approach on atypical test data gets generally better in the combined training model, compared to the single-source training model. Finally, for every test data, the performance of the multi-source ensemble model is close to the best performance. To be more precise, the relative difference of its RMSE value is less than $7\%$ from the RMSE of the best model on the same test data. This is empirical evidence of the reliability of the multi-source ensemble model.

The option price prediction performance of multi-source models with approaches $\mathcal{A}_{HH}$ and $\mathcal{A}_{DS}$, summarized in Table \ref{tab:XGB performance}, is illustrated more effectively through a histogram of residuals, defined as the differences between predicted and actual normalized price of options, as shown in Figure \ref{fig:histogram}. The horizontal axis represents the range of residuals, which is divided into 100 intervals, while the vertical axis shows the relative frequency. The histograms in Figure \ref{fig:histogram} are constructed using predictions from the XGBoost model, trained on combined datasets, and tested on NIFTY$50$ (plot \ref{fig:N50-hist}) and BANKNIFTY (plot \ref{fig:BNF-hist}) datasets, respectively. These predictions are made using both the $\mathcal{A}_{DS}$ and $\mathcal{A}_{HH}$ approaches, under typical and atypical testing scenarios. Hence, each plot has four histograms.
\begin{figure}[h]
\centering
\begin{subfigure}{.49\textwidth}
\centering
   \includegraphics[width=\linewidth]{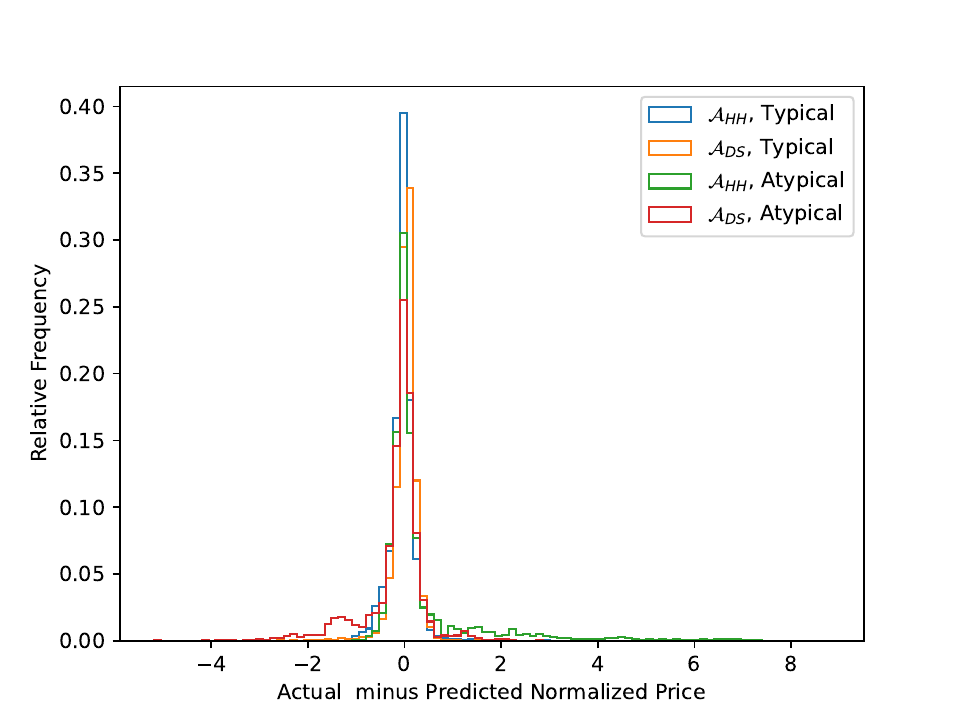}
    \caption{NIFTY50 testing data}\label{fig:N50-hist}
\end{subfigure}
\begin{subfigure}{.49\textwidth}
\centering
\includegraphics[width=\linewidth]{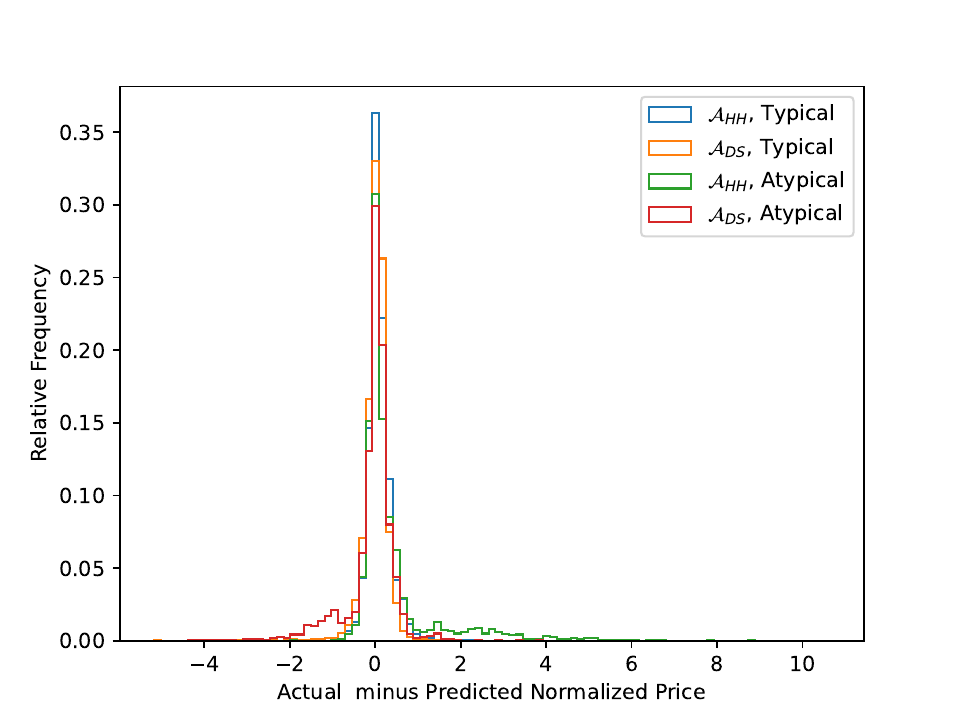}
      \caption{BANKNIFTY testing data}\label{fig:BNF-hist}
\end{subfigure}
    \caption{Histograms of subtraction of predicted from actual `normalized price' are plotted for models trained on combined N$50$ and BNF training data. Both models with $\mathcal{A}_{HH}$, and $\mathcal{A}_{DS}$ approaches are tested on different indices with both typical and atypical scenarios.}\label{fig:histogram}
\end{figure}
The performance reported in Table \ref{tab:XGB performance} indicates that under typical testing scenarios, the $\mathcal{A}_{HH}$ approach yields better accuracy than $\mathcal{A}_{DS}$ approach, as corroborated by the histogram. Specifically in plot \ref{fig:N50-hist}, for N$50$ test set, the residuals from $\mathcal{A}_{HH}$ approach (represented by the `blue' bars) are more tightly clustered around zero compared to the $\mathcal{A}_{DS}$ approach (`orange' bars). In contrast, for the BANKNIFTY test data, the difference in frequency levels between the two approaches is less pronounced, indicating a similar error margin. On the N50 typical test data, residual distributions for both the approaches exhibit (Figure \ref{fig:N50-hist}) symmetric and low dispersion around zero. This confirms the absence of bias in the prediction. A similar low dispersion of the prediction residual is also visible for BANKNIFTY typical test data in Figure \ref{fig:BNF-hist}. 

The dispersion of prediction residual for atypical test data appears to be significantly larger than that of the typical data scenario for every approach in both Figures \ref{fig:N50-hist} and \ref{fig:BNF-hist}. 
Additionally, an intriguing observation in the atypical testing scenario is that the $\mathcal{A}_{DS}$ approach leads to some outliers of overestimation of option prices relative to actual values, as evidenced by the leftward skew of the `red bars', while the $\mathcal{A}_{HH}$  approach leads to outliers of underestimation of prices, as indicated by the rightward skew of the `green' bars across both atypical testing indices. Hence the ensemble model performs significantly better on atypical data by neutralizing such polar outliers in $\mathcal{A}_{HH}$ and $\mathcal{A}_{DS}$ predictions.
In summary, the histograms provide a clearer and deeper insight into the models' performance across different testing scenarios than the RMSE errors offer in Table \ref{tab:XGB performance}.

\subsection{Model performance on synthetic test data} Now we present the results of two experiments conducted to assess the performance of trained models on synthetic data. The use of synthetic data overcomes real-world data limitations for assessing domain adaptation in rare but large domain shifts. The synthetic data is simulated using Black-Scholes-Merton model of asset prices and the formula for theoretical option prices. The primary goal of these experiments is to assess how the prediction quality of the proposed approaches varies as the test data shifts from the training dataset. It is important to clarify that these experiments are not intended to evaluate the overall performance of the models but to gain a comparative insight into model performances in extreme cases. These do not give a reliable performance measure in a typical scenario as both underlying asset returns and the option prices of test data are derived under specific mathematical assumptions that may not accurately reflect real-world prices. 

We simulate Geometric Brownian motion (GBM) with a drift parameter $\mu = 0.1$ and vary the volatility parameter from $1\%$ to $30\%$ in increments of $1\%$. For each volatility level, we generate daily data of length $520$ days. This test set is supplemented with prices of several near-ATM option contracts, with time to maturity in $\{10, 25, 40\}$ and risk-free interest rate of $5\%$, calculated using the Black-Scholes-Merton formula. We then test each of the trained models for each variant of the test data and plot the predictions' RMSE against the volatility parameter. 
The analysis is centered on comparing the performance of the models with $\mathcal{A}_{DS}$, $\mathcal{A}_{HH}$, and ensemble model $\mathcal{A}_E$ approaches, and trained on various datasets. Given the nature of simulated test data, although the distinction between typical and atypical test scenarios is not direct, still all the synthetic data generated with $\sigma>20\%$ may be viewed as atypical.

\subsubsection{Experiment 1} 
\begin{figure}[h]
\centering
\begin{subfigure}{.49\textwidth}
\centering
\includegraphics[width=\linewidth]{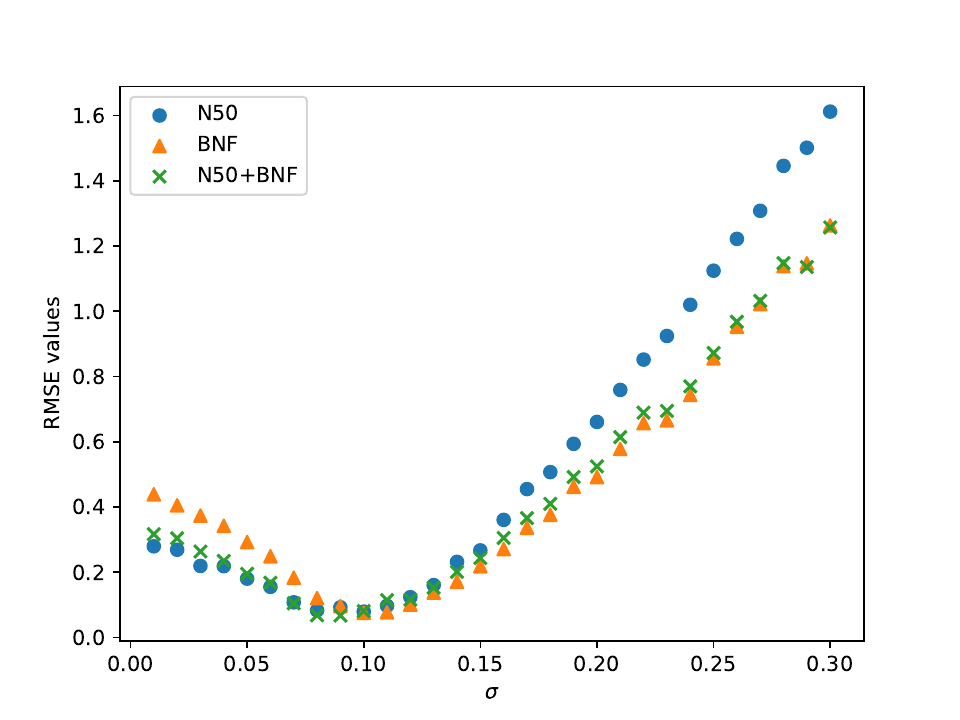}
\caption{Performance of $\mathcal{A}_{HH}$ models}\label{fig:no-cut-introspection}
\end{subfigure}
\begin{subfigure}{.49\textwidth}
\centering
   \includegraphics[width=\linewidth]{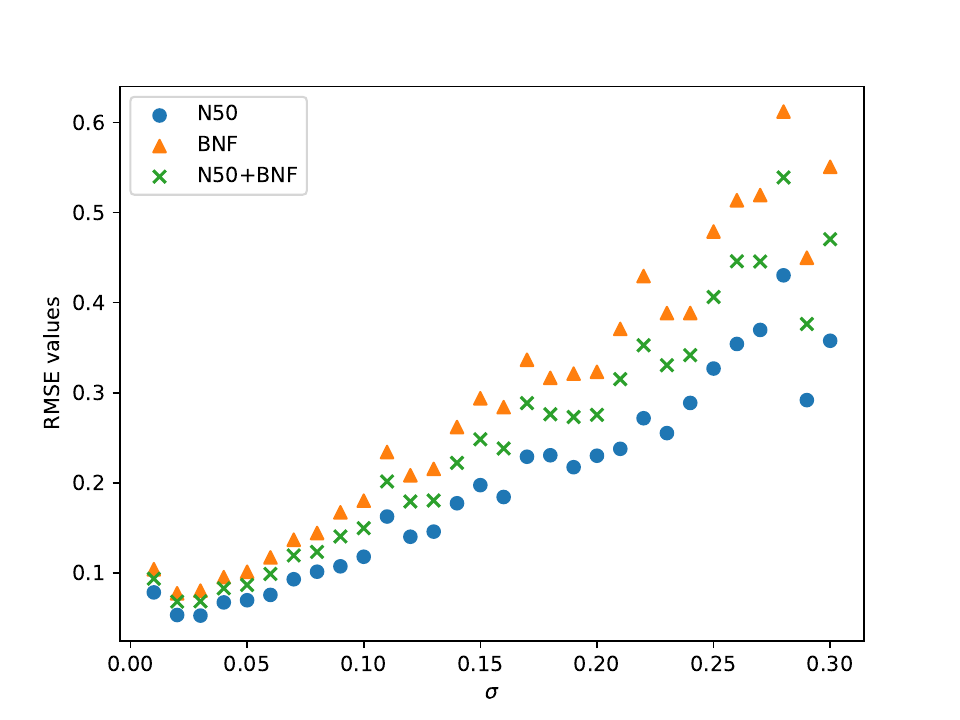}
\caption{Performance of $\mathcal{A}_{DS}$ models}\label{fig:cut-introspection}
\end{subfigure}
    \caption{Model performance on a range of parametrized synthetic data}\label{fig:introspection}
\end{figure}
In the first experiment, we evaluate and compare the models' prediction quality on synthetic test data for various training datasets. We perform this comparison for synthetic test data having varying volatility levels, and models with each approach, $\mathcal{A}_{HH}$ and $\mathcal{A}_{DS}$. 
Indeed, for every synthetic data with volatility parameter $\sigma$ not less than $12\%$, among $\mathcal{A}_{HH}$ models, the BNF trained model has the least RMSE, whereas N50 trained model has the largest RMSE. This can be explained using the historical volatility values of respective training data. As the historical volatility of BNF $(14.05\%)$ is higher than that of N50 $(11.38\%)$, the test data with high $\sigma$ is less unfamiliar to BNF-trained data than to N50-trained data. 
The same consistent comparison is not visible in the $\mathcal{A}_{DS}$ based models. 
In contrast, the $\mathcal{A}_{DS}$ approach aims to mitigate such errors. Figure \ref{fig:cut-introspection} provides supporting evidence of the absence of performance dependence on the type or volume of training data. 
Moreover,  as noted in \cite{GRT}, the effectiveness of approach $\mathcal{A}_{HH}$ diminishes with domain shift, i.e., the RMSE values of each $\mathcal{A}_{HH}$ model steadily increase for synthetic test data with larger volatility values. 
Although the RMSE of prediction by $\mathcal{A}_{DS}$ grows with larger $\sigma$ beyond  $\sigma =3\%$, the growth is slower and oscillatory than that by $\mathcal{A}_{HH}$. This growth of RMSE for $\mathcal{A}_{DS}$ prediction may be attributed to the growth of relative error in approximation relation \eqref{eqpsi} with the increase of disparity(max-min ratio) of volatilities as illustrated in Figure \ref{fig:scatter}.
These results align with the intent of making the $\mathcal{A}_{DS}$ approach less sensitive to discrepancies between the return distributions of the training and testing datasets.

\subsubsection{Experiment 2} 
\begin{figure}[h]
\centering
\begin{subfigure}{.49\textwidth}
\centering
   \includegraphics[width=\linewidth]{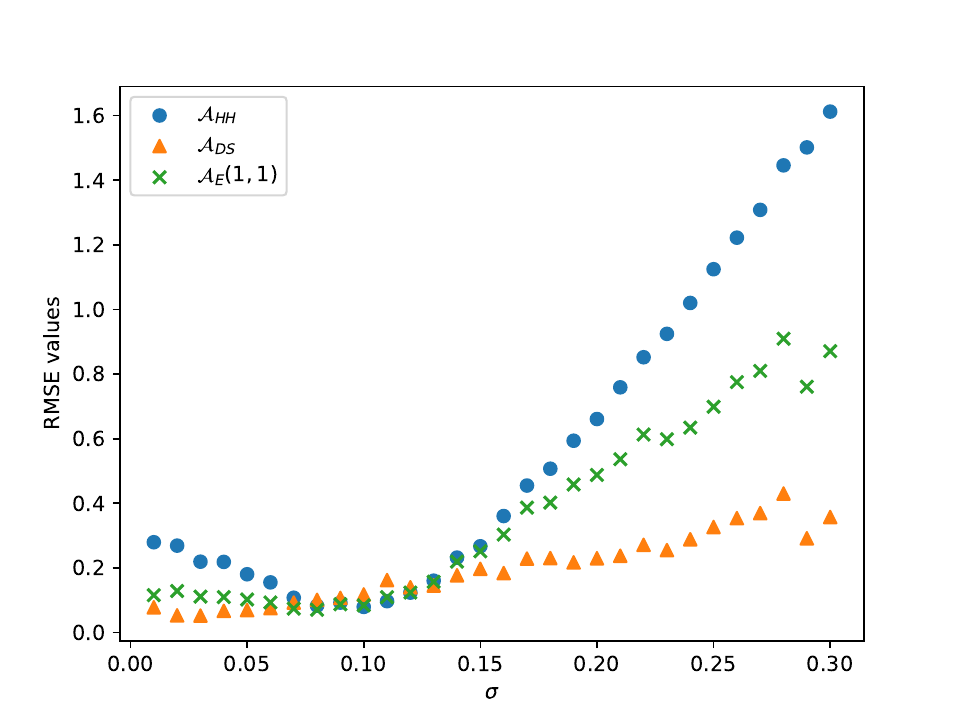}
    \caption{Performance of N50 trained model}\label{fig:N50-introspection}
\end{subfigure}
\begin{subfigure}{.49\textwidth}
\centering
\includegraphics[width=\linewidth]{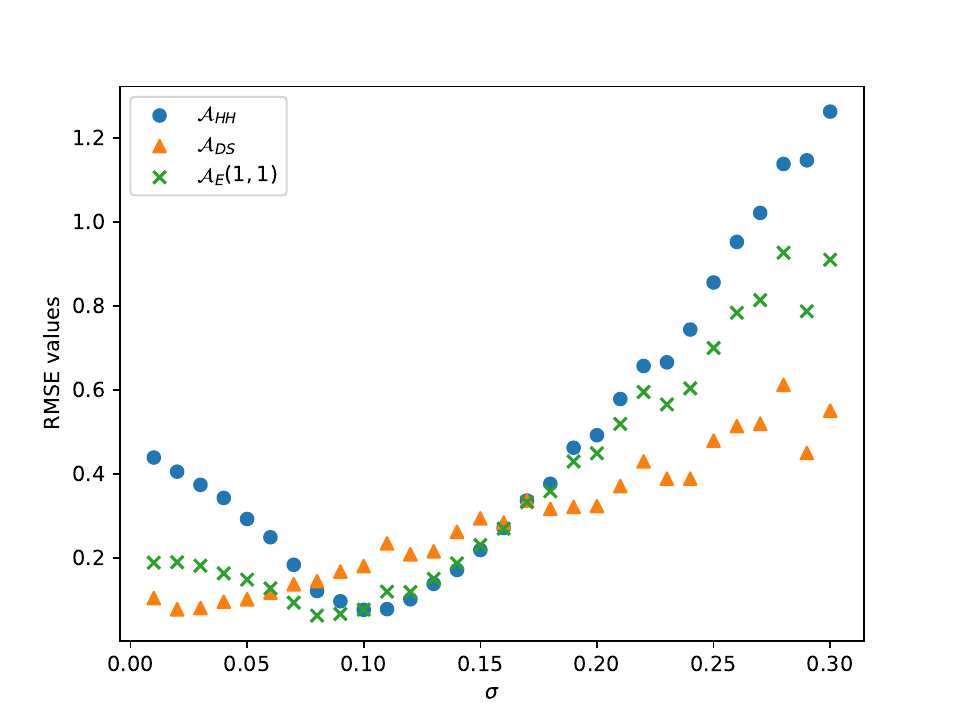}
      \caption{Performance of BNF trained model}\label{fig:BNF-introspection}
\end{subfigure}
\begin{subfigure}{.49\textwidth}
\centering
   \includegraphics[width=\linewidth]{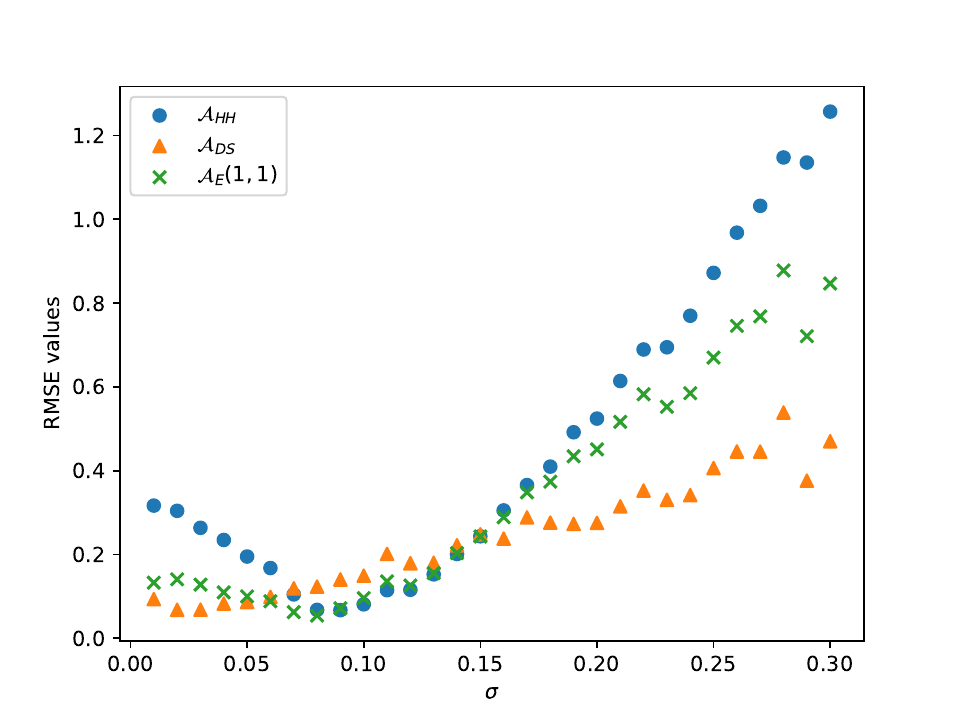}
    \caption{Performance of multi-source trained model}\label{fig:Comb-introspection}
\end{subfigure}
    \caption{Comparison of approaches $\mathcal{A}_{HH}$, $\mathcal{A}_{DS}$ and $\mathcal{A}_E$ for each of three different trained set using a parametric family of synthetic data} \label{fig:introspection-different dataset}
\end{figure}
In another experiment, we compare the RMSE values for all three approaches $\mathcal{A}_{HH}$, $\mathcal{A}_{DS}$ and the emsemble model on the afore mentioned collection of synthetic data. Separate comparison plots are presented for models trained on the NIFTY$50$, BANKNIFTY and combined datasets. As anticipated, above a volatility level (i.e., above $15\%$) of the test data, the approach $\mathcal{A}_{DS}$ outperforms $\mathcal{A}_{HH}$ for every training set (see each subplot of Fig \ref{fig:introspection-different dataset}). Additionally, performance graphs for the ensemble model are presented with  $\lambda_1=1$ and $\lambda_2=1$. In every case (i.e., subplots \ref{fig:N50-introspection}, \ref{fig:BNF-introspection}, \ref{fig:Comb-introspection}) the RMSE of the ensemble model is mostly between the RMSEs of $\mathcal{A}_{HH}$ and $\mathcal{A}_{DS}$ models, except for a small range of the volatility values of the synthetic data, where it is smallest. However, that small range corresponds to the typical scenario, as it is in the neighborhood of the HV of the training dataset. So, it explains why Table \ref{tab:XGB performance} conveys the best performance by $\mathcal{A}_{E(1,1)}$ models for real test datasets.

\section{Conclusion}\label{sec:conclusion}
This work presents several interesting elements and fresh ideas for option pricing applications of machine learning. This paper puts forward an interesting idea of utilizing a formula for implied volatility estimates to create a common representation space so that a better domain adaptation is achieved. This method also allows training models on several options markets (i.e., options with different underlying securities), since this could improve out-of-sample performance in each considered market. Finally, we formulate a novel enseble model, depending on a domain shift quotient, that works for both typical and atypical test data. We believe, this is a promising avenue for further research and could potentially bring some new insights into the field of data-driven option pricing. Moreover, we also assess the models using a family of synthetic test data for gaining a better understanding of the models.

\noindent \textbf{Acknowledgments:}
The first author acknowledges the financial support from the National Board for Higher Mathematics (NBHM), Department of Atomic Energy, Government of India, under grant number 02011-32-2023-R$\&$D-II-13347 and the Department of Science and Technology under the Project-Related Personal Exchange Grant DST/INTDAAD/P-12/2020. The support and the resources provided by ‘PARAM Brahma Facility’ under the National Supercomputing Mission, Government of India at the Indian Institute of Science Education and Research (IISER) Pune are gratefully acknowledged. The financial support by the London Mathematical Society through Scheme 5 ‘Collaborations with Developing Countries’, with Ref 52335, is greatly acknowledged by the last author.

\noindent \textbf{Declarations:}
\begin{itemize}
\item Conflict of interest: The authors declare no conflict of interest.
\item Availability of data and materials: Data sharing is not applicable to this article as no new data were created or analyzed in this study. All the empirical studies reported in this article are based on openly available data.
\item Authors' contributions: All authors contributed equally towards the manuscript.
\end{itemize}

  \bibliography{ml.bib}{}

@article {BS,
    AUTHOR = {Black, Fischer and  Myron Scholes},
     TITLE = {The Pricing of Options and Corporate Liabilities},
  JOURNAL = {Journal of Political Economy},
    VOLUME = {81},
      YEAR = {1973},
    NUMBER = {3},
     PAGES = {637–654},
       DOI = {10.1086/260062}
}

@article {Fengler2009,
    AUTHOR = {Fengler, Matthias R.},
     TITLE = {Arbitrage-free smoothing of the implied volatility surface},
   JOURNAL = {Quant. Finance},
  FJOURNAL = {Quantitative Finance},
    VOLUME = {9},
      YEAR = {2009},
    NUMBER = {4},
     PAGES = {417--428},
      ISSN = {1469-7688,1469-7696},
       DOI = {10.1080/14697680802595585},
       URL = {https://doi.org/10.1080/14697680802595585},
}

@Article{Malliaris1993,
author={Malliaris, Mary
and Salchenberger, Linda},
title={A neural network model for estimating option prices},
journal={Applied Intelligence},
year={1993},
month={Sep},
day={01},
volume={3},
number={3},
pages={193-206},
issn={1573-7497},
doi={10.1007/BF00871937},
url={https://doi.org/10.1007/BF00871937}
}

@article{Hutch94,
 URL = {http://www.jstor.org/stable/2329209},
 author = {James M. Hutchinson and Andrew W. Lo and Tomaso Poggio},
 journal = {The Journal of Finance},
 number = {3},
 pages = {851--889},
 publisher = {[American Finance Association, Wiley]},
 title = {A Nonparametric Approach to Pricing and Hedging Derivative Securities Via Learning Networks},
 volume = {49},
 year = {1994}
}

@INPROCEEDINGS{Maddala96,
  author={Qi, M and Maddala, GS},
  booktitle={Proceedings of 3rd International Conference on Neural Networks in the Capital Markets}, 
  title = {Option pricing using artificial neural networks: The case of S\&P 500 Index Call Options},  year={1996},
 publisher = {World Scientific},
 pages = {78--91}, 
  doi={10.1109/ICNN.1995.487522}
}

@article{Malliaris1996,
title = {Using neural networks to forecast the S{\&}P 100 implied volatility},
author = {Mary Malliaris and Linda Salchenberger},
journal = {Neurocomputing},
volume = {10},
number = {2},
pages = {183-195},
year = {1996},
note = {Financial Applications, Part I},
issn = {0925-2312},
doi = {https://doi.org/10.1016/0925-2312(95)00019-4},
url = {https://www.sciencedirect.com/science/article/pii/0925231295000194},
}

@INPROCEEDINGS{Boek95,
  author={Boek, C. and Lajbcygier, P. and Palaniswami, M. and Flitman, A.},
  booktitle={Proceedings of ICNN'95 - International Conference on Neural Networks}, 
  title={A hybrid neural network approach to the pricing of options}, 
  year={1995},
  volume={2},
  number={},
  pages={813-817 vol.2},
  doi={10.1109/ICNN.1995.487522}
}

@ARTICLE{Yao2000,
title = {Option price forecasting using neural networks},
author = {Yao, Jingtao and Li, Yili and Tan, Chew Lim},
year = {2000},
journal = {Omega},
volume = {28},
number = {4},
pages = {455-466},
url = {https://EconPapers.repec.org/RePEc:eee:jomega:v:28:y:2000:i:4:p:455-466}
}

@ARTICLE{Gradojevic2009,
  title    = "Option pricing with modular neural networks",
  author   = "Gradojevic, Nikola and Gen{\c c}ay, Ramazan and Kukolj, Dragan",
  journal  = "IEEE Trans Neural Netw",
  volume   =  20,
  number   =  4,
  pages    = "626--637",
  month    =  mar,
  year     =  2009
}

@article{GRT,
author = {Goswami, A. and Rajani, S. and Tanksale, A.},
title = {Data-driven option pricing using single and multi-asset supervised learning},
journal = {International Journal of Financial Engineering},
volume = {08},
number = {02},
pages = {2141001},
year = {2021},
doi = {10.1142/S2424786321410012},
URL = {         https://doi.org/10.1142/S2424786321410012
},
eprint = {         https://doi.org/10.1142/S2424786321410012
}
}

@article {Li05,
    AUTHOR = {Li, S.},
     TITLE = {A new formula for computing implied volatility},
   JOURNAL = {Appl. Math. Comput.},
  FJOURNAL = {Applied Mathematics and Computation},
    VOLUME = {170},
      YEAR = {2005},
    NUMBER = {1},
     PAGES = {611--625},
      ISSN = {0096-3003},
   MRCLASS = {91B28 (65M99)},
  MRNUMBER = {2177565},
       DOI = {10.1016/j.amc.2004.12.034},
       URL = {https://doi.org/10.1016/j.amc.2004.12.034},
}

@article{GARCIA200093,
title = {Pricing and hedging derivative securities with neural networks and a homogeneity hint},
journal = {Journal of Econometrics},
volume = {94},
number = {1},
pages = {93-115},
year = {2000},
issn = {0304-4076},
doi = {https://doi.org/10.1016/S0304-4076(99)00018-4},
url = {https://www.sciencedirect.com/science/article/pii/S0304407699000184},
author = {René Garcia and Ramazan Gençay}
}

@article{SVS24,
title = {Enhancing Option Pricing Accuracy in the {I}ndian Market: A {CNN}-{B}i{LSTM} Approach},
journal = {Computational Economics},
volume = {},
number = {},
pages = {},
year = {2024},
issn = {},
doi = {https://doi.org/10.1007/s10614-024-10689-z},
author = {Sharma, A. and  Verma, C.K. and Singh, P.},
}

@article{KTP23,
author = {Shubham, Kumar and Tiwari, Vivek and Patel, Kuldip Singh},
title = {Predictive Learning Methods to Price European Options Using Ensemble Model and Multi-asset Data},
journal = {International Journal on Artificial Intelligence Tools},
volume = {32},
number = {07},
pages = {2350034},
year = {2023},
doi = {10.1142/S0218213023500343},
URL = {https://doi.org/10.1142/S0218213023500343},
eprint = {https://doi.org/10.1142/S0218213023500343}
}

@inproceedings{CG16,
author = {Chen, T. and Guestrin, C.},
title = {{XGB}oost: A Scalable Tree Boosting System},
year = {2016},
isbn = {9781450342322},
publisher = {Association for Computing Machinery},
address = {New York, NY, USA},
url = {https://doi.org/10.1145/2939672.2939785},
doi = {10.1145/2939672.2939785},
booktitle = {Proceedings of the 22nd ACM SIGKDD International Conference on Knowledge Discovery and Data Mining},
pages = {785–794},
numpages = {10},
location = {San Francisco, California, USA},
series = {KDD '16}
}

@book {KK00,
    AUTHOR = {Kallianpur, G. and Karandikar, R. L.},
     TITLE = {Introduction to option pricing theory},
 PUBLISHER = {Birkh\"{a}user Boston, Inc., Boston, MA},
      YEAR = {2000},
     PAGES = {x+269},
      ISBN = {0-8176-4108-4},
   MRCLASS = {91B28 (60G99)},
  MRNUMBER = {1718056},
MRREVIEWER = {Christian Hipp},
       DOI = {10.1007/978-1-4612-0511-1},
       URL = {https://doi.org/10.1007/978-1-4612-0511-1},
}

@incollection{Bharadia96,
title = {Computing the {B}lack-{S}choles implied volatility},
editor = {Phelim P. Boyle et al.},
booktitle = {Advances in Futures and Options Research},
publisher = {JAI Press},
pages = {15-29},
year = {1996},
isbn = {1559388528},
author = {M.A. Bharadia and N. Christofides and G.R. Salkin}
}

@article{WR2022,
author = {Weiguan Wang and Johannes Ruf},
title = {A note on spurious model selection},
journal = {Quantitative Finance},
volume = {22},
number = {10},
pages = {1797--1800},
year = {2022},
publisher = {Routledge},
doi = {10.1080/14697688.2022.2097120},
URL = { 
 https://doi.org/10.1080/14697688.2022.2097120
},
eprint = { 
https://doi.org/10.1080/14697688.2022.2097120}
}

@ARTICLE{RW2020,
title = {Neural networks for option pricing and hedging: a literature review},
author = {Ruf, J. and Wang, W.},
year = {2020},
journal = {Journal of Computational Finance},
volume = {24},
number = {1},
pages = {1-46},
url = {https://www.risk.net/journal-of-computational-finance/7659611/neural-networks-for-option-pricing-and-hedging-a-literature-review}
}

@article {CLZ2021,
    AUTHOR = {Cao, Yi and Liu, Xiaoquan and Zhai, Jia},
     TITLE = {Option valuation under no-arbitrage constraints with neural
              networks},
   JOURNAL = {European J. Oper. Res.},
  FJOURNAL = {European Journal of Operational Research},
    VOLUME = {293},
      YEAR = {2021},
    NUMBER = {1},
     PAGES = {361--374},
      ISSN = {0377-2217,1872-6860},
       DOI = {10.1016/j.ejor.2020.12.003},
       URL = {https://doi.org/10.1016/j.ejor.2020.12.003},
}

@article {MR3991069,
    AUTHOR = {Sirignano, Justin and Cont, Rama},
     TITLE = {Universal features of price formation in financial markets:
              perspectives from deep learning},
   JOURNAL = {Quant. Finance},
  FJOURNAL = {Quantitative Finance},
    VOLUME = {19},
      YEAR = {2019},
    NUMBER = {9},
     PAGES = {1449--1459},
      ISSN = {1469-7688,1469-7696},
       DOI = {10.1080/14697688.2019.1622295},
       URL = {https://doi.org/10.1080/14697688.2019.1622295},
}
  \bibliographystyle{siam}

\end{document}